\documentclass[journal]{IEEEtran}
\usepackage{float}
\usepackage{placeins}
\usepackage[space]{cite}
\usepackage[utf8]{inputenc}
\usepackage[T1]{fontenc}
\usepackage{amsmath}
\usepackage{mathtools}
\usepackage{array}
\usepackage{multirow}
\newcommand{\PreserveBackslash}[1]{\let\temp=\\#1\let\\=\temp}
\newcolumntype{C}[1]{>{\PreserveBackslash\centering}p{#1}}
\newcolumntype{R}[1]{>{\PreserveBackslash\raggedleft}p{#1}}
\newcolumntype{L}[1]{>{\PreserveBackslash\raggedright}p{#1}}

\usepackage[colorinlistoftodos]{todonotes}

\usepackage[noend]{algorithmic}
\usepackage[linesnumbered,ruled,vlined,boxed,commentsnumbered]{algorithm2e}
\usepackage{amsmath}
\usepackage{pgfplots}
\usepackage{graphicx} 
\usepackage{subfigure}
\usepackage{multirow}
\usepackage{paralist}
\usepackage{wrapfig}
\usepackage{url}
\usepackage{amsfonts}
\usepackage{amssymb}
\usepackage{tikz}
\usetikzlibrary{scopes,patterns,intersections,calc}
\usepackage{cancel}
\usepackage{bm}
\usepackage{nicefrac}
\usepackage{color}
\usepackage{diagbox}
\usepackage{enumerate}
\usepackage{booktabs}
\usepackage{colortbl}
\usepackage{tabularx}
\usepackage{epstopdf}
\usepackage{diagbox}

\newtheorem{exmp}{\textbf{Example}}

\newtheorem{definition}{\textbf{Definition}}

\newtheorem{theorem}{\textbf{Theorem}}
\newtheorem{lemma}{\textbf{Lemma}}
\newtheorem{corollary}{\textbf{Corollary}}
\newtheorem{hyp}{\textbf{Hypothesis}}
\newcommand{\remove}[1]{}

\newenvironment{proof}{\noindent {\bf
		Proof.}}{\rule{3mm}{3mm}\par\medskip}

\pgfplotsset{compat=1.3}
\newcommand\MyLBrace[2]{%
  \left.\rule{0pt}{#1}\right\}\text{#2}}
  \definecolor{nvb}{RGB}{65,105,225}
\newcommand*\circled[1]{\tikz[baseline=(char.base)]{
            \node[shape=circle,draw,inner sep=1pt] (char) {#1};}}
            
\newcolumntype{P}[1]{>{\centering\arraybackslash}p{#1}}
\newcolumntype{M}[1]{>{\centering\arraybackslash}m{#1}}

\begin{document}
	
\title{How the Network Topology, Traffic Distribution, and Routing Scheme  Impact on the Spectrum Usage in Elastic Optical Networks}	

\author{Haitao~Wu,~Fen~Zhou,~\IEEEmembership{Senior~Member,~IEEE},~Zuqing~Zhu,~\IEEEmembership{Senior~Member,~IEEE},~Yaojun~Chen

\thanks{H. Wu is with the Department of Mathematics, Nanjing University, Nanjing 210093, China, and he is also with the CERI-LIA at the University of Avignon, France. (email: haitao.wu@alumni.univ-avignon.fr). }
\thanks{F. Zhou is with the CERI-LIA (Computer Science Laboratory) at the University of Avignon, France. (email: fen.zhou@univ-avignon.fr).}
\thanks{Z. Zhu is with the School of Information Science and Technology, University of Science and Technology of China, Hefei, Anhui 230027, China. (e-mail: zqzhu@ieee.org).}
\thanks{Y. Chen is with Department of Mathematics, Nanjing University, Nanjing 210093, China. (email: yaojunc@nju.edu.cn).}
}

\maketitle

\begin{abstract}
Elastic Optical Network (EON) has been considered as a promising optical networking technology to architect the next-generation backbone networks. Routing and Spectrum Assignment (RSA) is the fundamental problem in EONs to realize service provisioning. Generally, the RSA is solved by routing the requests with lightpaths first and then assigning spectrum resources to the lightpaths to optimize the spectrum usage. 
Thus, the spectrum assignment explicitly decide the final spectrum usage of EONs. However, besides the spectrum assignment, there are three other factors, the network topology, traffic distribution and routing scheme, implicitly impact on the spectrum usage. Few related work involves in the implicit impact mechanism. In this paper, we aim to provide a thoroughly theoretical analysis on the impact of the three key factors on the spectrum usage. To this end, two theoretical chains are proposed: (1) The optimal spectrum usage can be measured by the chromatic number of the conflict graph, which is positively correlated to the intersecting probability, \emph{i.e.}, the smaller the intersecting probability, the smaller the optimal spectrum usage; (2) The intersecting probability is decided by the network topology, traffic distribution and routing scheme via a quadratic programming parameterized with a matrix of conflict coefficients. The effectiveness of our theoretical analysis has been validated by extensive numerical results. Meanwhile, our theoretical deductions also permit to give several constant approximation ratios for RSA algorithms. 
\end{abstract}

\section{Introduction}
\label{sec:intro}

\IEEEPARstart{R}{ecently}, various types of connection requests proliferate rapidly in the backbone networks. Since the traffic bandwidths are growing exponentially while the spectrum resources in the optical layer are not unlimited, highly-efficient and flexible optical networking technologies have attracted intensive research interests \cite{Gerstel2012, s7:b1}. Specifically, due to fact that the flexibility of traditional Wavelength Division Multiplexing (WDM) networks is restricted by the fixed grids, the flexible-grid Elastic Optical Networks (EONs) have been considered as promising replacements to architect the next-generation backbone networks \cite{s6:b4, Zhu2013_JLT, s6:b5}. In an EON, the spectrum resources on each fiber link are divided into narrow-band (\emph{i.e.}, $12.5$ GHz or less) Frequency Slices (FSs), with which the EON can allocate just enough bandwidths to satisfy each connection request adaptively \cite{s7:b2,s7:b3}. Therefore, the spectrum utilization can be effectively improved in EONs in contrast to the WDM optical networks.

The fundamental problem to realize service provisioning in EONs is the Routing and Spectrum Assignment (RSA) \cite{s1:b4, Gong2012_CL}, \emph{i.e.}, how to route a connection request from its source to its destination  by a lightpath, and then assign a block of available FSs on it. The RSA problem has been proven to be $\mathcal{NP}$-hard in \cite{s1:b4}. Similar to Routing and Wavelength Assignment (RWA) problem in WDM networks\cite{ss4:b6}, two variants of the RSA problem have been studied in the literature: max-RSA and min-RSA\cite{ss4:b7,ss4:b8}. For the former problem, the focus is on the provisioning over an EON under limited spectrum resource, and the objective is to maximize the number of requests that can be served. The latter one has a planning concern and its objective is to minimize the required spectrum usage to serve all the requests\cite{ss4:b7,ss4:b8}. In this paper, we focus on the latter one with planning concern.

Although the spectrum assignment is the direct determinant, the network topology, traffic distributions and routing scheme implicitly impact on the final spectrum usage. After a set of requests, generated from the traffic distribution, arrive at the underlying EON and are routed on their lightpaths by the routing scheme, there exists an optimal spectrum assignment that reaches the optimal spectrum usage. From another perspective, the optimal spectrum usage is the 
lower bound on the final spectrum usage that can not be reduced no matter how the spectrum assignment optimizes. The optimal spectrum usage, \textit{i.e.}, the lower bound on the final spectrum usage, apparently stems from the network topology, traffic distributions and routing scheme.  
%\cite{s6:b1} such as: the fixed routing schemes ($\textit{e.g.}$, shortest paths or other $K$-shortest paths \cite{ss4:b1,ss4:b2}), and the alternate routing  schemes, which, based on some "optimal" principles such as "the least congested" \cite{ss4:b1} and "the smallest load" \cite{s1:b2} \textit{etc}, select the "optimal" path for each request in a set of predetermined candidate paths. 
While the RSA problem has been extensively studied \cite{ss4:b1,ss4:b2,s1:b2,s6:b1}, currently few related work involves in this topic.    
%On the one hand, given a same EON, under different traffic distributions, the optimal routing decisions are changed. On the other hand, given a same type of traffic distribution, under different topologies of EONs, the final performances are various.  
%How to take into account the traffic distribution network topology of the EON to make the optimal routing decision
Nonetheless, it is pivotal to comprehend how the network topology, traffic distribution and routing scheme impact on the optimal spectrum usage. The grasp on this implicit impact mechanism can help to adjust the network topology, traffic distribution and routing scheme to optimize the final spectrum usage so as to boost the performance of EONs.   
 
In this work, we deduce two theoretical chains to rigorously analyze the three key factors' impact on the optimal spectrum spectrum usage.

\subsection{The theoretical chain 1}
We prove that the optimal spectrum usage can be bounded by the chromatic number of conflict graph. More specifically, the conflict graph is an auxiliary graph that describes the intersections among the lightpaths for served requests. Several constant approximation ratios of RSA algorithms have also been derived based on the theoretical analysis. We then show that the intersecting probability determines the chromatic number of the conflict graph by leveraging an important lemma in random graph theory. This is the first theoretic chain that the intersecting probability, through the chromatic number, determines the optimal spectrum usage.

\subsection{The theoretical chain 2}
We propose a concept of conflict coefficient, which is a conditional intersecting probability determined by the network topology and traffic distribution. Through a matrix composed of the conflict coefficients, we introduce a quadratic programming named Global Optimal Formulation (GOF) which is decided by the routing scheme and decides the intersecting probability. This is the second theoretical chain that the network topology, traffic distribution and routing scheme, through the GOF, determine the intersecting probability.

The main contributions of this work are summarized as follows:
\begin{itemize}
\item We give the upper and lower bounds of the optimal spectrum usage by analyzing the chromatic number of the conflict graph, which is a non-trivial extension of the counterpart in the WDM networks. Several constant approximation ratios of RSA algorithms have also been derived through the theoretical analysis.
\item By leveraging random graph theory, we provide an analytical approach on how to connect the chromatic number of conflict graph with the intersecting probability. Meanwhile, a matrix of conflict coefficients and the GOF, which embody the impact of the network topology, traffic distribution and routing scheme, are also proposed to determine the intersecting probability.
\item Within the proposed theoretical analysis, we evaluate three realistic EONs under two traffic distributions by the conflict coefficients and GOFs. Extensive simulations have also been conducted, whose results verify the effectivenesses of our theoretical deductions. 
%\item We prove that the spectrum assignment can be solved optimally when the conflict graph of RSA is a tree, and provide a spectrum assignment algorithm with guaranteed approximation ratio for mesh EONs.
\end{itemize}

The remaining of this paper is organized as follows. Section \ref{sec:rw} introduces the related work and our motivation. We present the formulation of the RSA problem and the conflict graph in Section \ref{sec:form}. Then, the theoretical work that reveals the relation between the optimal spectrum usage and the chromatic number of conflict graph is discussed in Section \ref{sec:themec}. In Section \ref{sec: tomufiandip}, we introduce the connection between the intersecting probability and the chromatic number, and propose the conflict coefficients and the quadratic programming GOF. Within the proposed theoretical analysis, we evaluate three EONs under two traffic distributions in Section \ref{sec:ciceon}. Extensive simulations under different scenarios are conducted in Section \ref{sec:numres} to verify our analysis. Finally, Section \ref{sec:conclusion} summarizes this paper.

\section{Related Work and Motivation}
\label{sec:rw}
In order to support future Internet cost-efficiently, optical networking technologies have to be optimized towards efficiency, flexibility, and scalability. By allocating spectrum resources in terms of narrow-band FSs \cite{s6:b4,s7:b4}, EONs can realize both sub-wavelength and super-channel optical transmissions according to the bandwidth requirements of connection requests. Therefore, EONs have been considered as promising candidates to architect the next-generation backbone networks. Parallel to RWA problem in WDM networks \cite{sss:b1}, the RSA problem is fundamental to service provisioning in EONs \cite{s1:b4, Gong2012_CL}. However, the RSA problem is more challenging, since it has to deal with various channel sizes and take into account the contiguity of FSs while the RWA does not. 

The RSA problem is $\mathcal{NP}$-hard \cite{s1:b4} and can be naturally separated into two subproblems\cite{s6:b1}: lightpath routing and spectrum assignment. In the lightpath routing, a request should be routed on an appropriate  lightpath to connect its source to its destination. The spectrum assignment, relatively analogous to the graph coloring problem \cite{htwuTON2017}, is to allocate available contiguous and continued FSs on the lightpath. Although the spectrum assignment directly decides the final spectrum usage, the routing scheme is also critical to the final spectrum usage. Meanwhile, the network topology and traffic distribution should also be taken in to account when planning the routing scheme. Since under different circumstances, the strategy of the routing scheme should be adaptive. Many routing schemes have been proposed \cite{ss4:b1,ss4:b2,s1:b2}. These routing schemes are executed in a relatively common way: (1) First pre-calculate a set of candidate paths for each source-destination pair; (2) then select the "best" path to the request based on some heuristic principles such as "the least congested" \cite{ss4:b1} and "the smallest load" \cite{s1:b2}. A comprehensive survey can be found in \cite{s6:b1}. 

Although many theoretical works \cite{s1:b4,ss4:b9,ss4:b10} on the RSA problem have been proposed by tools such as Integer Linear Programming (ILP) or mixed ILP. 
However, currently there is no theoretical work to interpret the implicit impact of the network topology, traffic distribution and routing scheme on the spectrum usage. Meanwhile, the ultimate goal of this research is to determine the optimal routing scheme while taking into account the network topology and traffic distribution. To the best of our knowledge, there has been no previous work on this topic. In this work, by leveraging random graph theory, we give a theoretical explanation about the implicit impact of the three factors on the spectrum usage.

\begin{table}[!h]
\caption{Notations} \label{tab: notations}
\begin{tabular}{|p{1cm}p{7cm}|}
\hline
$G(V,E)$     &The underlying EON, where $V$ is the set of nodes, and $E$ is the set of directed fiber links.\\
$\mathcal{D}$     &The traffic distribution, which specially refers to the distribution of occurrence probabilities of source-destination pairs.\\
$w_{sd}$     &The occurrence probability of source-destination pair $(s,d)$ determined by the traffic distribution $\mathcal{D}$.\\
$\mathcal{R}$       & The set of connection requests in $G(V,E)$. \\
$n$          & $=|\mathcal{R}|$, the number of connection requests. \\
$R_i(s_i,d_i)$  & $R_i \in \mathcal{R}$ represents the $i$-th connection request, where $s_i, d_i \in V$ are the source and destination nodes respectively generated by $\mathcal{D}$.\\
$\mathbb{N^+}$    &The set of positive natural numbers representing the FS index set in the spectrum domain lying in each directed fiber link $e \in E$. \\
$\alpha$ & $\alpha \in \mathbb{N^+}$ is the upper bound of the bandwidths required by requests. \\
$\beta$ & $\beta \in \mathbb{N^+}$ is the lower bound of the bandwidths required by requests. \\
$R_i^w$ & The number of contiguous FSs (bandwidth requirement) required by $R_i$, which is in the range of $[\alpha, \beta]$.\\
$\mathcal{P}_i$ & The set of all the directed paths from $s_i$ to $d_i$ in $G(V,E)$.\\
$P_i$ & $P_i \in \mathcal{P}_i$ is the directed path on which $R_i$ is routed.\\
$W_i$ &  The set of contiguous FSs assigned to $R_i$.\\
$R_i^b$ & $R_i^b \in \mathbb{N^+}$ is the start-index of $W_i$.\\
$R_i^a$ & $R_i^a \in \mathbb{N^+}$ is the end-index of $W_i$.\\
$GB$ & $GB \in \mathbb{N^+}$ is the number of FSs required for the guard-band.\\
MUFI &  $=\max_{s \in (\mathop{\cup} W_i)} (s)$, is the maximum used FS index.\\
\hline
$\hat{G}(\hat{V},\hat{E})$ & The conflict graph which is a weighted undirected graph, where $\hat{V}$ is the vertex set corresponding to $\mathcal{R}$, and $\hat{E}$ is the edge set.\\
$\hat{v}_i$ & $\hat{v}_i \in \hat{V}$ corresponds to $R_i$.\\
$\hat{v}^w_i$ & $=R_i^w$, the vertex weight of $\hat{v}_i$.\\
$W_{\hat{v}_i}$ &  The set of contiguous FSs assigned to $\hat{v}_i$.\\
$\hat{v}^b_i$ & $\hat{v}^b_i \in \mathbb{N^+}$ is the start-index of $W_{\hat{v}_i}$.\\
$\hat{v}^a_i$ & $\hat{v}^a_i \in \mathbb{N^+}$ is the end-index of $W_{\hat{v}_i}$.\\

$\hat{v}^w_{I(i)}$ & The $i$-th biggest vertex weight in $\hat{V}$ \textit{i.e.}, $\hat{v}^w_{I(1)} \geq \hat{v}^w_{I(2)} \geq ... \geq \hat{v}^w_{I(n)}$.\\
$\hat{v}^w_{D(i)}$ & The $i$-th smallest vertex weight in $\hat{V}$ \textit{i.e.}, $\hat{v}^w_{D(1)} \leq \hat{v}^w_{D(2)} \leq ... \leq \hat{v}^w_{D(n)}$.\\

$opt(\hat{G})$ & The optimal spectrum assignment for $\hat{G}(\hat{V},\hat{E})$.\\
$|opt(\hat{G})|$ & The MUFI of $opt(\hat{G})$, which is the optimal one.\\
 $\chi(\hat{G})$ & The chromatic number of $\hat{G}$.\\
\hline
$p$ & The intersecting probability of any two requests.  \\
$K$ & The number of candidate paths for each source-destination pair.  \\
$p_i$ & The probability for a request to be routed on the $i$-th shortest candidate path. \\ 
$\theta_{ij}$ &  The conflict coefficient that represents the intersecting probability of any two requests under the condition that one is routed on the $i$-th candidate path and the other on the $j$-th path of their own source-destination pairs.\\
 CM & $= [\theta_{ij}]$, the real symmetric Conflict Matrix composed of all $\theta_{ij}$.\\
\hline
\end{tabular}
\end{table}

\section{Network Model and Problem Description}
\label{sec:form}
% In this section, we will briefly introduce the architecture of EONs, the constraints and the objective function of RSA problem.
In this section, we present the network model considered in this paper and the formulation of the RSA. Some necessary notations are summarized in Table \ref{tab: notations}.

\subsection{Network Model}

\subsubsection{Network Topology}
We use a directed graph $G(V,E)$ to represent the topology of an EON, where $V$ and $E$ denote the sets of nodes and directed fiber links respectively as in Fig. \ref{fig:zxfg}. A set of FSs lies on each directed fiber link as in Fig. \ref{fig:fg}. 

\begin{figure}[!htb]
\centering
 \includegraphics[height=0.11\textwidth]{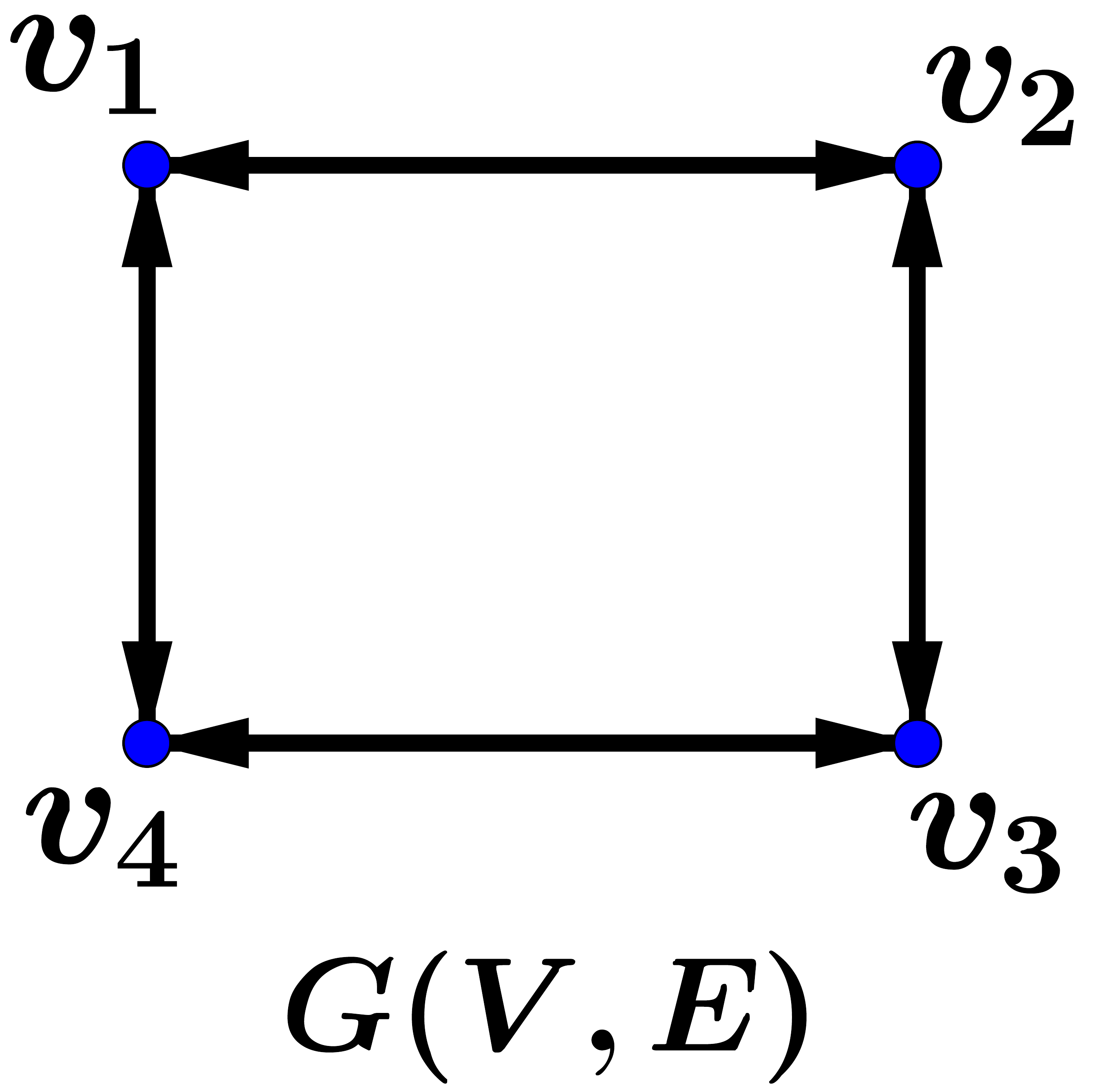}
\caption{An EON of 4 nodes and 4 bidirectional fiber links.}
\label{fig:zxfg}
\end{figure}

\begin{figure}[!htb]
\centering
 \includegraphics[height=0.13\textwidth]{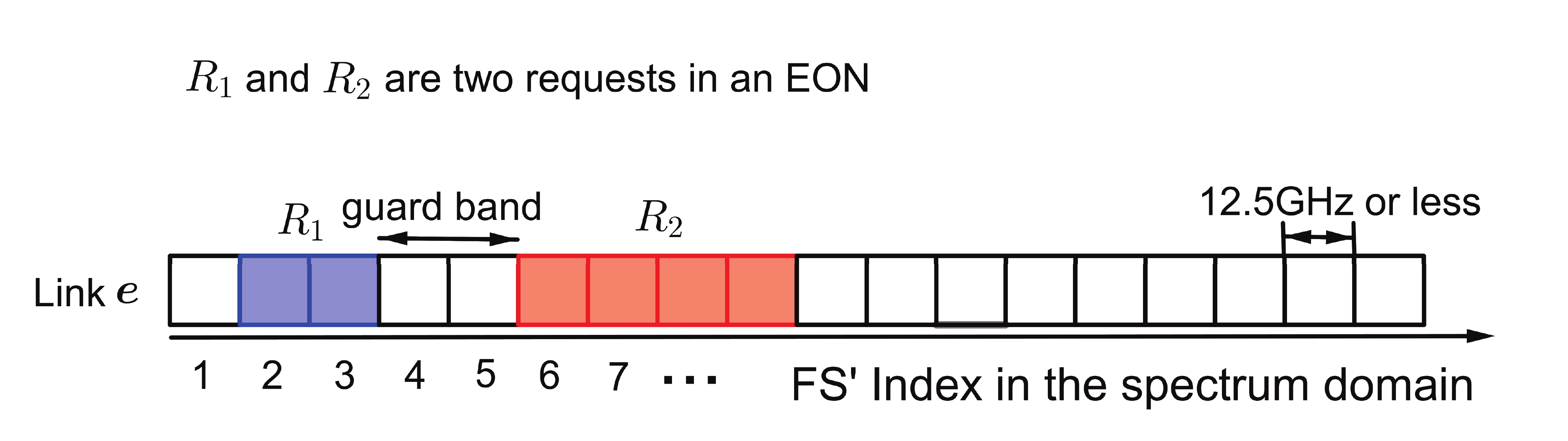}
\caption{FSs and guard-bands in a directed fiber link of an EON.}
\label{fig:fg}
\end{figure}

\subsubsection{Traffic Distribution}
A request $R$ in the EON $G(V,E)$ is a source-destination pair $(s \in V,d \in V)$ together with a bandwidth requirement $R^w$. Given an EON $G(V,E)$, there are $|V| \times (|V|-1)$ kinds of source-destination pairs. 
In this paper, we treat each request $R$ as a random variable whose  $(s,d)$ is generated from a traffic distribution $\mathcal{D}$. The $\mathcal{D}$ specially refers to the distribution of occurrence probabilities of the $|V| \times (|V|-1)$ source-destination pairs. In a backbone EON $G(V,E)$, the traffic distribution $\mathcal{D}$ usually has some statistical characteristics in a period. For example,
assuming that request $R$ is generated from the traffic distribution $\mathcal{D}$ in Table \ref{tab: TD}, the $(s,d)$ of the $R$ has $45\%$ probability to be $(v_1,v_3$) or $(v_3,v_1)$ and $1\%$ to be others, which implies that $v_1$ and $v_3$ are two import data center nodes in the EON.

\begin{table}[!htb]
\centering
  \caption{Traffic Distribution $\mathcal{D}$}
\label{tab: TD}
\resizebox{0.47\textwidth}{!}{\begin{tabular}{|C{1.2cm}|C{1.8cm}|C{1.2cm}|C{1.8cm}|}
\hline
Node Pair ($s$, $d$)  &  Occurrence Probability  & Node Pair ($s$, $d$) & Occurrence Probability \\
 \hline
 $(v_1,v_2)$   & $w_{v_1v_2}=1\%$ & $(v_2,v_1)$ & $w_{v_2v_1}=1\%$ \\
 \hline
$(v_1,v_3)$ & $w_{v_1v_3}=45\%$ & $(v_3,v_1)$ & $w_{v_3v_1}=45\%$ \\
\hline
$(v_1,v_4)$ & $w_{v_1v_4}=1\%$ & $(v_4,v_1)$ & $w_{v_4v_1}=1\%$\\
\hline
$(v_2,v_3)$ & $w_{v_2v_3}=1\%$ & $(v_3,v_2)$ & $w_{v_3v_2}=1\%$\\
\hline
$(v_2,v_4)$ & $w_{v_2v_4}=1\%$ & $(v_4,v_2)$ & $w_{v_4v_2}=1\%$\\
\hline
$(v_3,v_4)$ & $w_{v_3v_4}=1\%$ & $(v_4,v_3)$ & $w_{v_4v_3}=1\%$\\
\hline
\end{tabular}}
\end{table}

In this paper, we assume that there is a set of $n$ requests $\mathcal{R}=\{R_1,R_2,...,R_n\}$ to be served, and each $R_i(s_i,d_i)$ is generated from a given traffic distribution $\mathcal{D}$.

\subsubsection{Routing Scheme}

In general, the number of paths connecting a source-destination pair is exponential in an EON. It is impossible to enumerate all possible paths. Thus, a practical way is to pre-compute a set of $K$ ($K$ is a preset constant) shortest candidate paths for each source-destination pair by some $K$-shortest path algorithms, and select one from these candidate paths to connect the source-destination pair of the request \cite{ss4:b1,ss4:b2,s1:b2}. The most import thing for a routing scheme is how to appropriately dispense these $K$ candidate paths to these source-destination pairs of requests. We use $p_i$ to denote the percentage of requests which are routed on the $i$-th shortest candidate paths of their own source-destination pairs, and thus $ \sum^K_{i=1} p_i=1$. From the perspective of probability, since each request $R(s,d)$ is treated as a random variable in this paper, each $p_i$ can also be interpreted as the probability of the request routed on the $i$-th shortest candidate path (Since $\sum^K_{i=1} p_i=1$, this interpretation is well-defined). The proportion array $(p_1,p_2,...,p_K)$ is determined by the routing scheme.  Thus, in this paper, the routing scheme is represented by the array $(p_1,p_2,...,p_K)$. For example, the array $(1,0,...,0)$ represents the routing scheme that select the first shortest candidate paths to connect all source-destination pairs of requests, which corresponds to the shortest-path routing scheme in \cite{s6:b1}.

%In this paper, we consider a set of unicast communication requests $\mathcal{R}$, and assume that the spectrum resources in each fiber link are enough to serve all the requests (\textit{i.e.}, no blocking). 

% and the FS assigned to the request should be contiguous in the spectrum domain. Two FS sets assigned to two different requests whose lightpaths share common fiber links must be separated by a guard-band (\textit{i.e.}, a number of FS) to mitigate the mutual interference as shown in Fig. \ref{fig:fg}.

\subsection{Problem Description}
In this paper, we study the RSA problem with a planning concern, \emph{i.e.}, given a set of connection requests $\mathcal{R}$ in an EON $G(V,E)$, we intend to minimize the spectrum usage in the optical fibers. For each request $R_i(s_i,d_i) \in \mathcal{R}$, we need to establish a lightpath and assign enough bandwidths on it so as to forward the data of the request. More specially, the RSA problem consists in selecting a lightpath $P_i$ from the set $\mathcal{P}_i$ for $R_i$ and assigning just enough FS set $W_i$ on this lightpath while satisfying the following constraints:
\begin{itemize}
\item \textbf{Bandwidth Requirement Constraint}. The number of FSs assigned to each request should no smaller than its bandwidth requirement, \textit{i.e.}, the cardinality of $W_i$ assigned to $R_i$ must be equal to its weight:
	\begin{equation}\label{eqn: brc}
	|W_i| = R_i^w, \quad \forall R_i \in \mathcal{R}.
	\end{equation}
	
Without loss of generality, we make the following assumption in this paper.
	\begin{hyp}
		\label{hyp:rg}
		The bandwidth requirement of each request is uniformly distributed in the range of $[\alpha, \beta]$, \textit{i.e.}, $R^w_i \in [\alpha, \beta]~~\forall i$,  where $\alpha$ and $\beta$ are two constants, \textit{e.g.}, $[\alpha, \beta]=[1, 2]$ in \cite{s1:b2} or $[1, 3]$ in \cite{s2:b1}.
	\end{hyp}
	\item \textbf{Spectrum Contiguity Constraint}. The FSs assigned to request $R_i$ must be contiguous in the spectrum domain, \textit{i.e.}, $\mathbb{N^+}$. Thus, $W_i$ can be expressed as $\{R_i^b,R_i^b+1,...,R_i^a-1,R_i^a\}$. This is a physical layer constraint for all-optical communications.
   \item \textbf{Spectrum Continuity Constraint}. 
   All the directed fiber links on the lightpath for $R_i$ (\textit{i.e.}, $e \in P_i$) should  be assigned with the same set of contiguous FSs $W_i$.
   \item \textbf{Guard Band Constraint}.
To mitigate mutual interference, when $P_i$ and $P_j$ share common some directed fiber link(s), the distance between $W_i$ and $W_j$ in the spectrum domain should be no less than $GB$ (as shown in Fig. \ref{fig:fg}):
\begin{equation}\label{eqn: dr}
	distance(W_i,W_j) \geq GB, \quad \forall  P_i \cap P_j \neq \emptyset,
\end{equation}
where, $$distance(W_i,W_j)=\min
\limits_{\mbox{$\begin{array}{c}
s \in W_i, t \in W_j\end{array}$}} \left(|s-t|-1\right).$$
\end{itemize}

%\footnote{In the study of optical network, there are usual two variants of the objective function in the literature: 1. to minimize the spectrum resources while routing all requests; 2. to maximize the number of routed requests while fixing the spectrum resources\cite{ss4:b6}. In this paper, we focus on investigating the former one.}

For planning purposes, the RSA problem studied in this paper aims to minimize the \textbf{Maximum Used FS Index} (MUFI), which can be expressed by Eq. (\ref{eqn: dsa}):
%\textbf{Objective Function:}
\begin{align}
\label{eqn: dsa}
%\textrm{Objective Function: }
\min \left[\max_{s \in (\mathop{\cup} W_i)} (s) \right] \qquad \qquad (\textbf{RSA}).
\end{align}

Given an EON $G$, a traffic distribution $\mathcal{D}$ and a routing scheme $(p_1,p_2,...,p_K)$, the studied objective of this paper is to figure out how the three factors impact on the optimal MUFI, $\textit{i.e.}$, the lower bound on the final spectrum usage as mentioned in the introduction.

\subsection{Conflict Graph}
After the set of requests $\mathcal{R}=\{R_1,R_2,...,R_n\}$, generated from the traffic distribution $\mathcal{D}$, arrive at the EON $G(V,E)$ and are routed on their lightpaths by the routing scheme $(p_1,p_2,...,p_K)$, we can construct an auxiliary graph, conflict graph, which embodies the three factors' impact on the optimal MUFI. The formal definition is given as follows.

% \newpage

% \begin{definition}
% \label{def:cg}
% The \textit{Triangles-Connecting Sequence} $(T_1,...,T_{i-1},T_i,T_{i+1}...,T_n)$ is such graph constituted by a sequence of triangles $T_i$, where $1\leq i \leq n$, that $T_i$ has two vertices, one shared with $T_{i-1}$ and one with $T_{i+1}$. Especially, when $n=1$, it is just a trivial triangle.
% \end{definition}

% \newpage

\begin{definition}
\label{def:cg}
The \textit{conflict graph} \cite{sss:b1} $\hat{G}(\hat{V},\hat{E})$ is such \textbf{a weighted undirected graph} whose vertex set $\hat{V}$ represents the set of requests, \textit{i.e.}, $\mathcal{R}$. Any two vertices $\hat{v}_i, \hat{v}_j \in \hat{V}$ (representing $R_i$ and $R_j$ respectively) are connected by an edge $\hat{e} \in \hat{E}$, \textit{i.e.}, they are adjacent in $\hat{G}$, if and only if $P_i$ \textbf{intersects with} $P_j$, \textit{i.e.}, $P_i \cap P_j \neq \emptyset$ \textbf{(at least one directed fiber link shared by $P_i$ and $P_j$)}, where $P_i$ and $P_j$ are the lightpaths selected for $R_i$ and $R_j$ respectively. We denote by $\hat{v}^w_i$ the weight of vertex $\hat{v}_i$, and $\hat{v}^w_i=R^w_i$. 

Besides, we can also define \textit{proper spectrum assignment for the conflict graph}: Let $\hat{v}^b_i$, $\hat{v}^a_i$ and $W_{\hat{v}_i}$ have the same meanings as $R^b_i$, $R^a_i$ and $W_i$ respectively, and $|W_{\hat{v}_i}|=\hat{v}^w_i$. If $\hat{v}_i$ and $\hat{v}_j$ are adjacent in $\hat{G}$, then the distance between $W_{\hat{v}_i}$ and $W_{\hat{v}_j}$ should be no less than $GB$. Then, these $W_{\hat{v}_i}$ compose of a proper spectrum assignment for the conflict graph.
\end{definition}

\begin{figure}[!htb]
\centering
 \begin{subfigure}[Routed lightpaths.]{
        \label{subfig:tslpp}
		\includegraphics[height=0.187\textwidth]{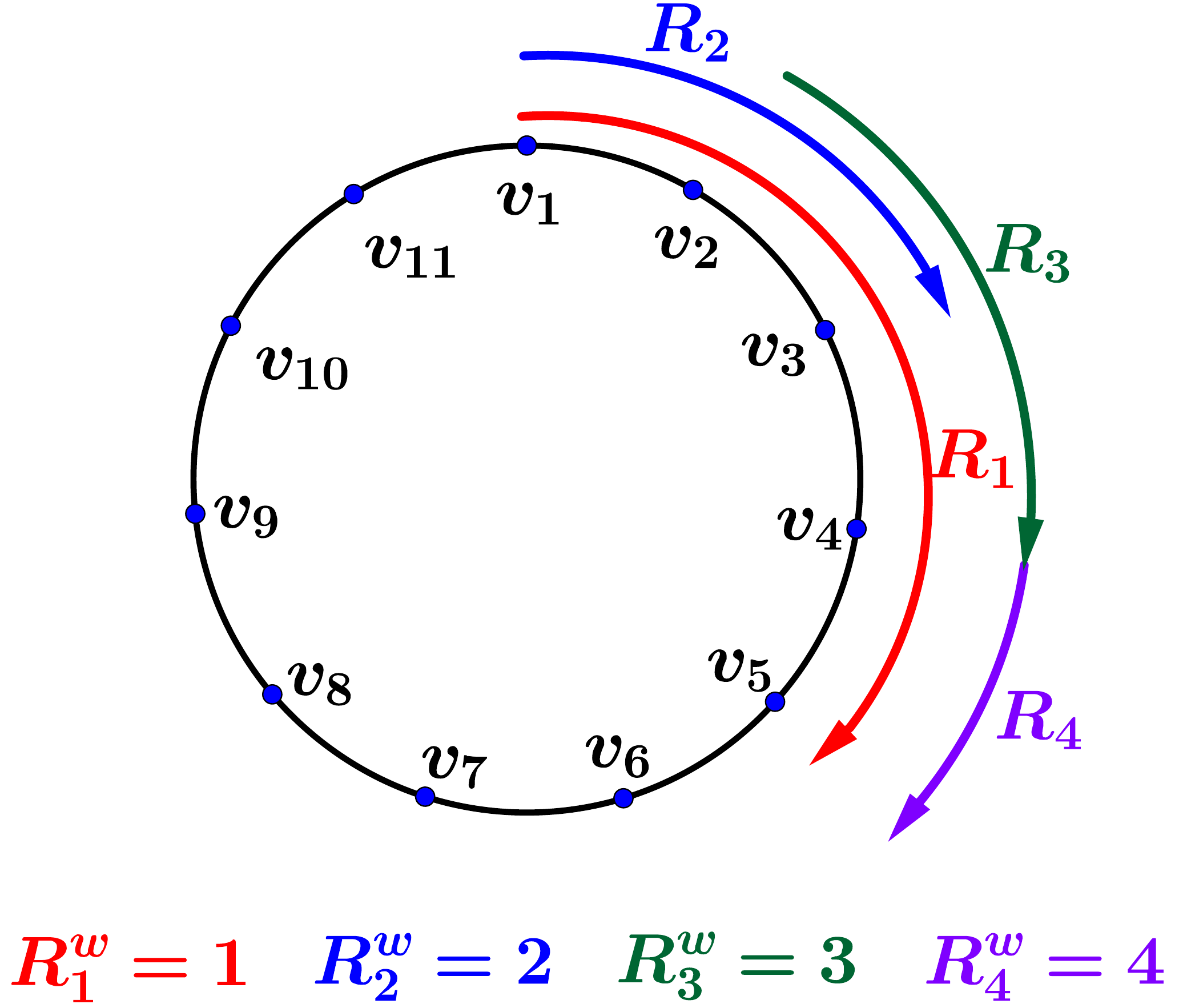}}
    \end{subfigure}
      \begin{subfigure}[Conflict graph.]{
        \label{subfig:tcgre}
		\includegraphics[height=0.19\textwidth]{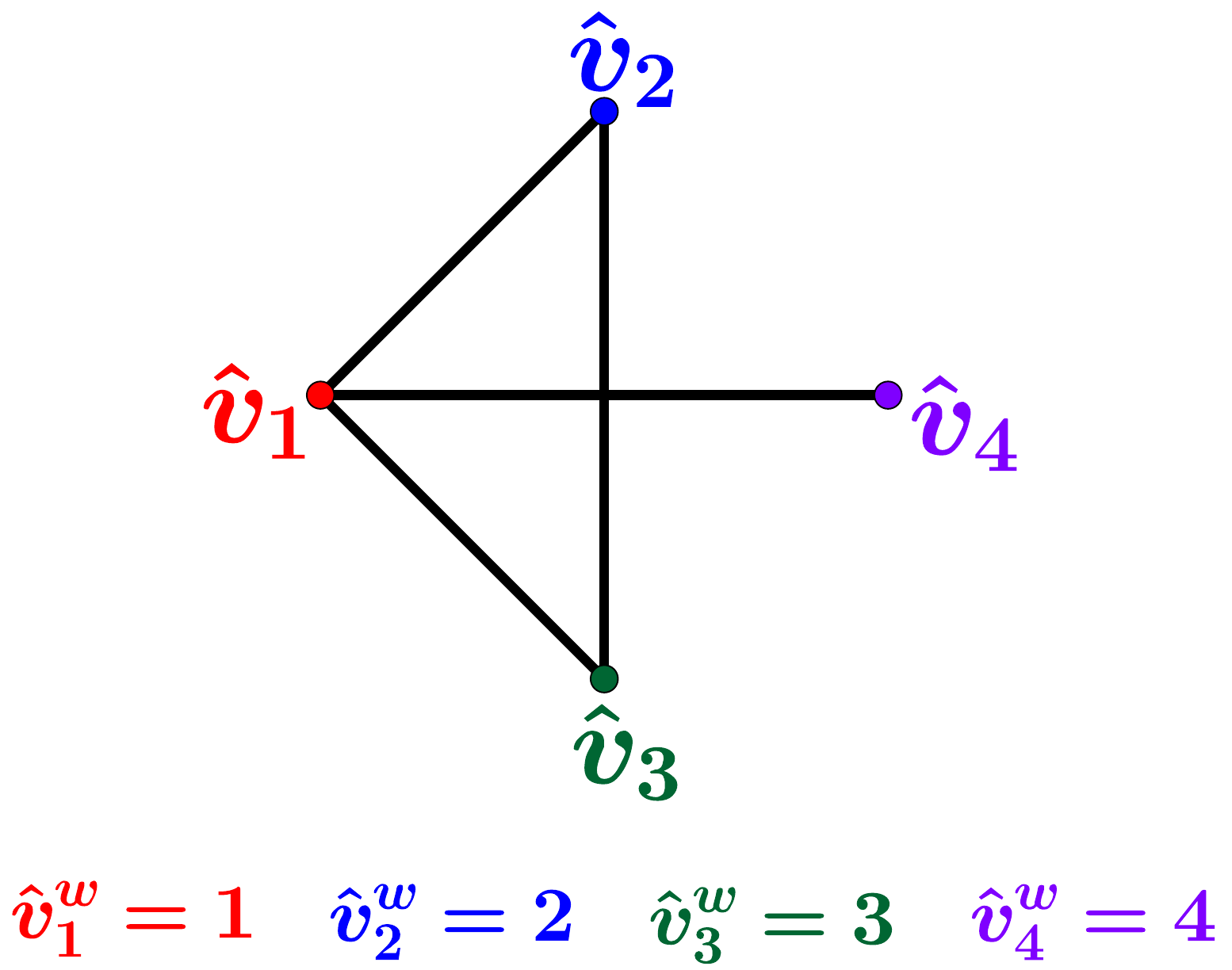}}
    \end{subfigure}
\caption{The conflict graph of the routed lightpaths.}
\label{fig:rfexam}
\end{figure}

Figure \ref{subfig:tcgre} showcases a 4-node conflict graph for the 4 requests (they are already routed on their lightpaths) in Fig. \ref{subfig:tslpp}, where $\hat{v}_i$ corresponds to $R_i, \forall 1\leq i\leq 4$. 

According to the definition, obviously, any proper spectrum assignment for the conflict graph corresponds to a proper spectrum assignment for these lightpaths, vice versa. Thus, the optimal spectrum assignment $opt(\hat{G})$ for the conflict graph, which  produces the minimum MUFI $|opt(\hat{G})|$, corresponds to the optimal spectrum assignment for these lightpaths. Therefore, the optimal MUFI is equal to $|opt(\hat{G})|$.

\section{Optimal MUFI and Chromatic Number of Conflict Graph}
\label{sec:themec}

In this section, we explore the relation between the $|opt(\hat{G})|$ and the chromatic number $\chi(\hat{G})$\footnote{The chromatic number is the minimum number of colors needed to color $\hat{G}$ such that adjacent vertices do not share a same color.} of  the conflict graph $\hat{G}$.

\subsection{Optimal MUFI and Chromatic Number}

In the WDM networks, the  parallel relation is that the minimum number of wavelengths used on the conflict graph is equal to its chromatic number. However, in EONs, how to determine the $|opt(\hat{G})|$ for the conflict graph $\hat{G}$ has not been investigated yet. Thus, we address this issue by giving Theorem \ref{the:chropt}.

\begin{theorem}
\label{the:chropt}
Assuming $\hat{G}(\hat{V},\hat{E})$ is the conflict graph, then
\begin{equation}
\footnotesize
\label{eqn: corebounds}
(\chi(\hat{G})-1)\cdot GB+\sum_{i=1}^{\chi(\hat{G})}\hat{v}^w_{D(i)} \leq |opt(\hat{G})|\leq (\chi(\hat{G})-1)\cdot GB+\sum_{i=1}^{\chi(\hat{G})}\hat{v}^w_{I(i)}.
\end{equation}
where $\hat{v}^w_{I(i)}$ and $\hat{v}^w_{D(i)}$ are the $i$-biggest and $i$-smallest vertex weight in $\hat{V}$ respectively, and $\chi(\hat{G})$ is the chromatic number of $\hat{G}$.
\end{theorem}

\begin{proof}
We prove the upper bound by finding a feasible spectrum assignment $f$ whose MUFI is not bigger than it. This solution $f$ can be obtained in such way as shown in Fig. \ref{fig:upper}: (1) Separating $\hat{V}$ into $\chi(\hat{G})$ disjoint independent sets\footnote{According to the definition of chromatic number, each monochrome part is an independent set, \textit{i.e.}, no two vertices of which are adjacent.} ; (2) assigning FSs for each independent set (The number of FSs assigned is equal to the biggest vertex weight in this set); (3) patching them up. It is easy to see the MUFI of $f$ is not bigger than the upper bound.

\begin{figure}[!htb]
\centering
 \includegraphics[height=0.3\textwidth]{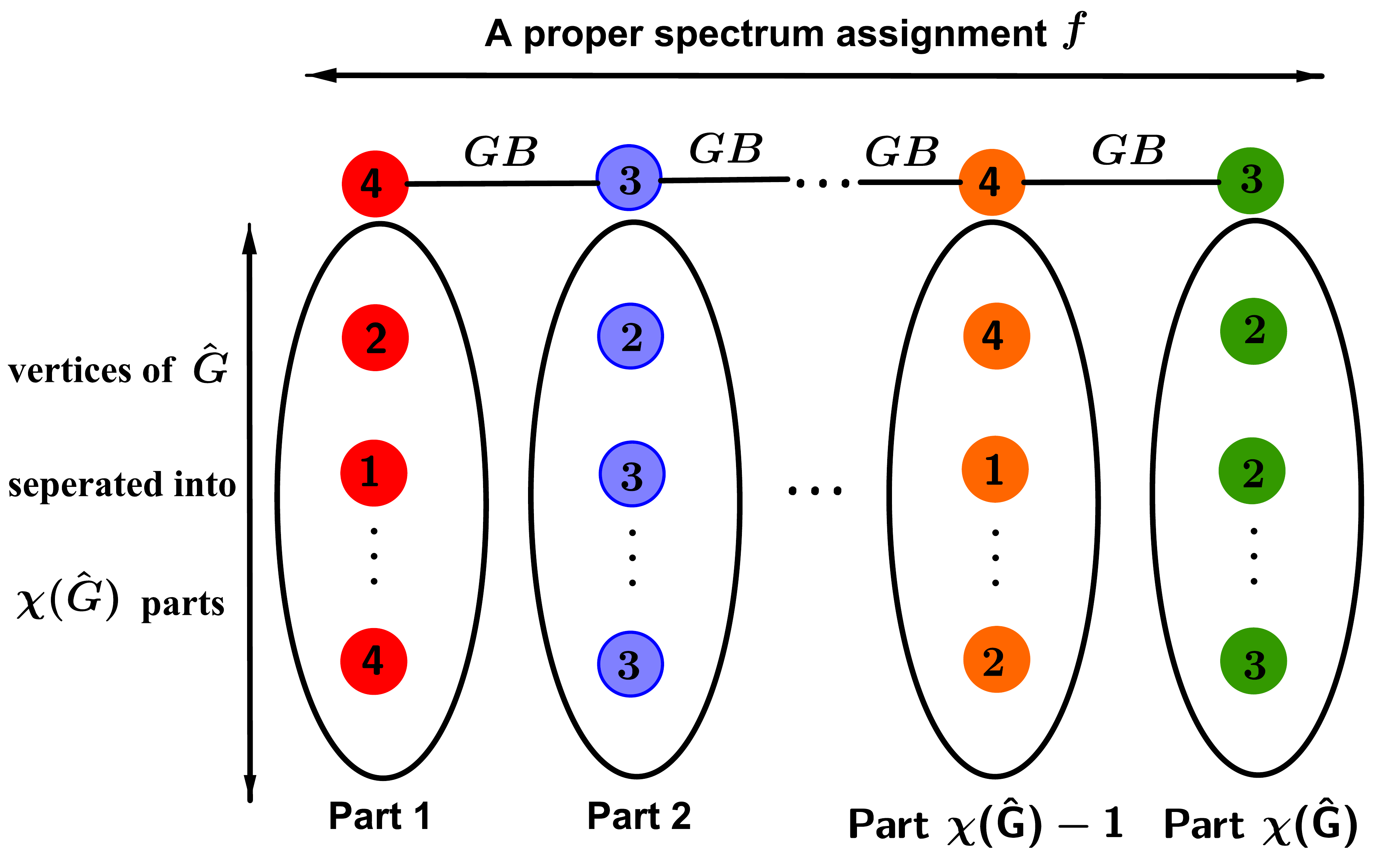}
\caption{A feasible spectrum assignment $f$ where each circle represents a vertex and the number in the circle is the corresponding vertex weight.}
\label{fig:upper}
\end{figure}

Now, we prove the lower bound. Let $opt(\hat{G})=\{W_{\hat{v}}, \forall \hat{v} \in  \hat{V}\}$  be the optimal spectrum assignment. With respect to $opt(\hat{G})$, let $A=\{\hat{v}^a, \forall \hat{v} \in \hat{V}\}$ and $B=\{\hat{v}^b, \forall \hat{v} \in \hat{V}\}$ be the optimal end-index and start-index sets for all $W_{\hat{v}}$ in $opt(\hat{G})$ respectively. We then can separate $\hat{V}$ into different parts as follows:

First, let, \textit{w.l.o.g.}, $\hat{v}_{1}$ be the vertex whose end-index $\hat{v}^a_1$ is the minimum in $A$, then we assert that $\mathcal{F}_1=\{\hat{v}|\hat{v}^b \leq \hat{v}^a_1+GB,\forall \hat{v} \in \hat{V} \}$ is an independent set of $\hat{G}$, \textit{i.e.}, for any two vertices $\hat{v}_i,\hat{v}_j \in \mathcal{F}_1$, $\hat{v}_i$ is not adjacent to $\hat{v}_j$ in $\hat{G}$. We prove it by contradiction. Assuming $\hat{v}_i$ is adjacent to $\hat{v}_j$ in $\hat{G}$ and $W_{\hat{v}_i}=[\hat{v}^b_i, \hat{v}^a_i]$ and $W_{\hat{v}_j}=[\hat{v}^b_j, \hat{v}^a_j]$, we have $distance(W_{\hat{v}_i},W_{\hat{v}_j}) \geq GB$. Thus, we, \textit{w.l.o.g.}, can assume $\hat{v}^a_i + GB < \hat{v}^b_j$. However, according to $\mathcal{F}_1$, we have $\hat{v}^b_j \leq \hat{v}^a_1+GB$ and $\hat{v}^a_i \geq \hat{v}^a_1$ (Since $\hat{v}^a_1$ is the minimum in all end-indices), which means  $\hat{v}^b_j \leq \hat{v}^a_i + GB$, a contradiction. Therefore $\mathcal{F}_1$ is an independent set of $\hat{G}$.

Next, let, \textit{w.l.o.g.}, $\hat{v}_2$ be the vertex whose end-index is the minimum in $\hat{V} \setminus  \mathcal{F}_1$, and similarly, we can assert $\mathcal{F}_2=\{\hat{v}|\hat{v}^b \leq \hat{v}^a_2+GB,\forall \hat{v} \in \hat{V}\setminus \mathcal{F}_1\}$ is an independent set of $\hat{G}$. After finite steps using the same technique, say $k$, we can separate $\hat{V}$ into $k$ disjoint independent sets: $\mathcal{F}_1$, $\mathcal{F}_2$,..., $\mathcal{F}_k$. Besides, according to the principle of selecting $\mathcal{F}_i, \forall 1 \leq i \leq k$, we have  $\hat{v}^a_i+GB < \hat{v}^b_{i+1}, \forall 1 \leq i \leq k$.  Therefore, we have $(k-1)\cdot GB+\sum_{i=1}^{k}\hat{v}^w_i \leq |opt(\hat{G})|$. Figure \ref{fig:lower} sketches the process above.

\begin{figure}[!htb]
\centering
 \includegraphics[height=0.085\textwidth]{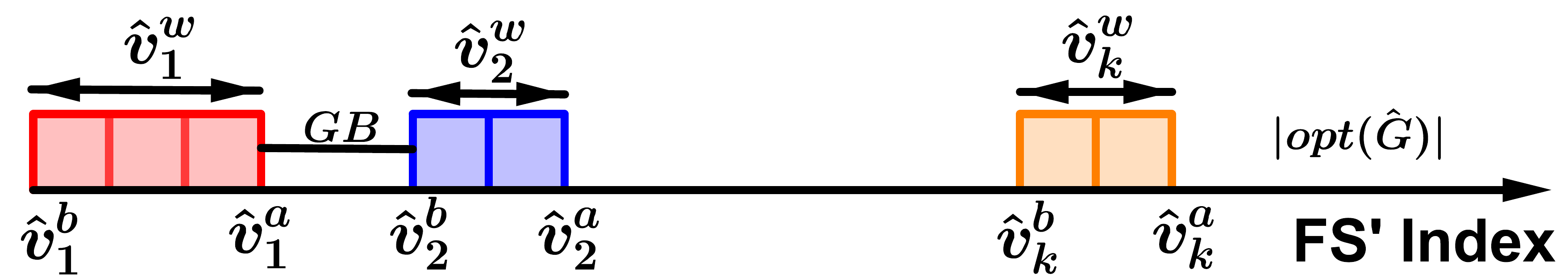}
\caption{The lower bound of $|opt(\hat{G})|$.}
\label{fig:lower}
\end{figure}

Finally, according to the definition of chromatic number that $\chi(G)$ is the minimum number of independent sets into which we can separate $\hat{V}$, so we immediately have $\chi(G) \leq k$. Besides $\sum_{i=1}^{k}\hat{v}^w_{D(i)} \leq \sum_{i=1}^{k}\hat{v}^w_i$, we hence have $(\chi(\hat{G})-1)\cdot GB+\sum_{i=1}^{\chi(\hat{G})}\hat{v}^w_{D(i)}  \leq  (k-1)\cdot GB+\sum_{i=1}^{k}\hat{v}^w_{D(i)}  \leq (k-1)\cdot GB+\sum_{i=1}^{k}\hat{v}^w_i \leq |opt(\hat{G})|$. 
Thus, we get the proof.
\end{proof}

Theorem \ref{the:chropt} reveals the relation between the chromatic number of the conflict graph $\hat{G}$ and the optimal MUFI $|opt(\hat{G})|$, which is a non-trivial generation from the RWA (Since in the RWA, the guard band constraint does not exists so $GB=0$, and each vertex weight is the same as a unit wavelength. Therefore, the lower bound is equal to the upper bound in Eq. (\ref{eqn: corebounds}) and the optimal number of wavelength is the chromatic number). Under Hypothesis \ref{hyp:rg} and Theorem \ref{the:chropt}, we can get Corollary \ref{cor: ulbs}.

\begin{corollary}
\label{cor: ulbs}
Assuming $\hat{G}(\hat{V},\hat{E})$ is the conflict graph and $\hat{v}^w \in [\alpha, \beta]$, then we have
\begin{equation}
	\footnotesize
	\label{eqnazer: corebounds}
(\chi(\hat{G})-1)\cdot GB+\chi(\hat{G})\cdot \alpha \leq |opt(\hat{G})|\leq (\chi(\hat{G})-1)\cdot GB+\chi(\hat{G})\cdot \beta.
\end{equation}

\end{corollary}

In practice, as $\hat{v}^w_{D(i)}$ and $\hat{v}^w_{I(n-i)}$ will become closer with the growth of $i$, the gap between the upper and lower bounds in Theorem \ref{the:chropt} will be smaller than that in Corollary \ref{cor: ulbs}. In other words, the $|opt(\hat{G})|$ is limited in a narrow interval bounded by the chromatic number $\chi(\hat{G})$. Thus, $\chi(\hat{G})$ directly determines  $|opt(\hat{G})|$.   

%Under the framework in Section \ref{subsecnn:cgorsaa}, the conflict graph of an RSA algorithm decides its optimality. From now on, we call $|opt(\hat{G})|$ as the optimal MUFI of an RSA algorithm, where $\hat{G}$ is the conflict graph of the RSA algorithm.  

\subsection{Some Results on Approximation Ratios}
After the lightpath routing, the conflict graph $\hat{G}(\hat{V},\hat{E})$ is constructed. Now, a natural question is how to obtain a good spectrum assignment to approach $|opt(\hat{G})|$. Since the $|opt(\hat{G})|$ is bounded by the chromatic number $\chi(\hat{G})$,  we can get an approximation algorithm (say $APX$) of spectrum assignment from a graph coloring algorithm (say $COL$) by slightly modifying it as follows: (1) Using $COL$ to partition the vertex set $\hat{V}$ into $k$ disjoint independent sets (each set in a monochrome), \textit{i.e.}, $C_1$, $C_2$,...,$C_k$, where $k$ is the chromatic number produced by $COL$; (2) using the same way that we obtain the upper bound in Fig. \ref{fig:upper} to assign FS sets. We get Corollary \ref{cor: apxi}.

% \begin{figure}[!htb]
% \centering
% \includegraphics[height=0.06\textwidth]{YYY1.eps}
% \caption{Coloring order.}
% \label{fig:tco}
% \end{figure}
% Thus, given a conflict graph $\hat{G}(\hat{V},\hat{E})$, a graph coloring algorithm said $A$ can be modified to an algorithm said $APX$ to optimize MUFI in such way: First, we utilize $A$ to separate $\hat{V}$ into several independent sets \textit{i.e.}, $C_1$, $C_2$,...,$C_k$. Then, we arrange $\hat{V}$ into the coloring order as shown in Fig. \ref{fig:tco}.
\begin{corollary}
\label{cor: apxi}
Given a conflict graph $\hat{G}(\hat{V}, \hat{E})$, if there is a polynomial algorithm $COL$ which can guarantee a $\rho$ approximation ratio for chromatic number, then there is a polynomial algorithm $APX$ which can guarantee a $\rho\cdot \max\{\frac{\beta}{\alpha},2\}$ approximation ratio for $|opt(\hat{G})|$, where $\alpha$ and $\beta$ are defined in Hypothesis \ref{hyp:rg}.
\end{corollary}

\begin{proof}
Since $\chi(\hat{G})=1$ is a trivial case, we just consider $\chi(\hat{G}) \geq 2$.
% Assuming that $A$ is a coloring algorithm which can guarantee $\rho$ ratio for $\chi(\hat{G})$ and $APX$ is the algorithm constructed as above mentioned to optimize $opt(\hat{G})$, we now prove that $APX$ can guarantee $\rho\cdot \max\{\frac{\beta}{\alpha},2\}$ ratio for $|opt(\hat{G})|$.
Supposing that $COL$ separates $\hat{V}$ into $C_1, C_2, ..., C_k$, let, \textit{w.l.o.g.}, $\hat{v}_1, \hat{v}_2, ..., \hat{v}_k$ be the vertices with the biggest vertex weight in each independent set respectively (just like that in Fig. \ref{fig:upper}). Then according to the construction of $APX$, we have $|APX(\hat{G})| \leq (k-1)\cdot GB+\sum_{i=1}^{k}\hat{v}^w_i$, where $|APX(\hat{G})|$ is the MUFI computed by $APX$. According to Corollary \ref{cor: ulbs} and the deduction above, we have $\dfrac{|APX(\hat{G})|}{|opt(\hat{G})|} \leq \dfrac{(k-1)\cdot GB+k\cdot \beta}{(\chi(\hat{G})-1)\cdot GB+\chi(\hat{G})\cdot \alpha} \leq \max \{\rho\cdot \frac{\beta}{\alpha}, \frac{k-1}{\chi(\hat{G})-1)}\}$. Since $\chi(\hat{G}) \geq 2$, $\frac{k-1}{\chi(\hat{G})-1} \leq 2\cdot \frac{k}{\chi(\hat{G})}=2\cdot \rho$. Hence, $\dfrac{|APX(\hat{G})|}{|opt(\hat{G})|} \leq \rho\cdot \max\{\frac{\beta}{\alpha},2\}$.
\end{proof}
For some special conflict graphs, their chromatic number can be solved in polynomial time. Here, we have Corollary \ref{cor: aprii}.
\begin{corollary}
\label{cor: aprii}
If the conflict graph $\hat{G}(\hat{V},\hat{E})$ is a perfect graph \footnote{Perfect graphs are an important graph class in graph theory which have many good properties such as their chromatic numbers can be polynomial-timely obtained.}, then there is a polynomial algorithm which can guarantee a $\frac{\beta}{\alpha}$ approximation ratio for $|opt(\hat{G})|$.
\end{corollary}

\begin{proof}
According to \cite{s3:b1}, the chromatic numbers of perfect graphs can be solved in polynomial time, \textit{i.e.}, there is an algorithm such that $\rho=1$. In the proof of Corollary \ref{cor: apxi}, we have $\dfrac{|APX(\hat{G})|}{|opt(\hat{G})|} \leq \max \{\rho\cdot \frac{\beta}{\alpha}, \frac{k-1}{\chi(\hat{G})-1)}\}$. By $\rho=\frac{k}{\chi(\hat{G})}=1$, and $\frac{\beta}{\alpha} \geq 1$,  we get this Corollary.
\end{proof}

In EONs with a tree topology, the conflict graphs over them are always perfect graphs. Then, we have:
\begin{corollary}
\label{cor: apriii}
If EON $G(V,E)$ is a tree, there is a polynomial algorithm which can guarantee a $\frac{\beta}{\alpha}$ ratio for the optimal MUFI.
\end{corollary}

\begin{proof}
Because the $G(V,E)$ is a tree, there is only one way to route each request, \textit{i.e.}, the conflict graph is unique. Besides, according to intersection graph theory \cite{s3:b2}, the conflict graphs over a tree network are always chordal graphs which belong to the perfect graph class. Therefore, the proof follows.
\end{proof}
However, in general cases, for the conflict graph $\hat{G}(\hat{V},\hat{E})$, it is intractable to approximate $\chi(\hat{G})$ within a ratio of $|\hat{V}|^{1-\epsilon}$ for any constant $\epsilon > 0$ \cite{s2:b3}. In other words, compared with $\chi(\hat{G})$, the final result could be very bad for any polynomial-time algorithm. Now, the best approximation ratio $\rho$ for $\chi(\hat{G})$ is $\mathcal{O}(|\hat{V}|\cdot (\log\log|\hat{V}|)^2/\log(|\hat{V}|)^3)$ \cite{s3:b4}.

From above discussions, we can see that the spectrum assignment, closely analogous to the graph coloring problem, is intractable, and the $|opt(\hat{G})|$ is determined by the conflict graph. Therefore, the three factors of the network topology, traffic distribution and routing scheme are critical for the final performance of the RSA, since they decide the property of the conflict graph $\hat{G}$.  How to reduce the chromatic number $\chi(\hat{G})$, thereby the $|opt(\hat{G})|$, is pivotal.

\section{Theoretical Chains of the Impact}
\label{sec: tomufiandip}

In the previous section, we theoretically deduced the impact of the $\chi(\hat{G})$ on the $|opt(\hat{G})|$. Now, the important thing is to study how the network topology $G$, traffic distribution $\mathcal{D}$ and routing scheme $(p_1,p_2,...,p_K)$ affect the $\chi(\hat{G})$.  

To this end, we need some preparations in random graph theory. We first introduce the concept of intersecting probability and its relation with the chromatic number. Then, we deduce how the intersecting probability can be influenced by the network topology $G$, traffic distribution $\mathcal{D}$ and routing scheme $(p_1,p_2,...,p_K)$. Thus, by the connection between the intersecting probability and the chromatic number, we finally figure out the impacts of the three factors on the optimal spectrum usage.

\subsection{Intersecting Probability and Chromatic Number}

%\subsection{Intersecting Probability}
%In a backbone EON $G(V,E)$, the traffic distribution $\mathcal{D}$ usually has some statistical characteristics in a period: For example, each source-destination pair in $G$ may have the same occurrence probability, or there are some important date-center nodes which constitute the majority of the connection requests. 

%The source and destination pair $(s_i,d_i)$ of a request $R_i$ is usually assumed to follow some distribution $\mathcal{D}$ in $V$, \textit{e.g.}, uniform distribution or real client distribution. 

%Under the random circumstance, the source and the destination of each request are generated by $\mathcal{D}$, and its bandwidth follows a uniform distribution in the range $[\alpha, \beta]$ by Hypothesis \ref{hyp:rg}. 

%In light of the fact that each request is viewed as a random variable, we introduce the intersecting probability of any two requests. 

\begin{definition}
In the EON $G(V,E)$, the \textit{intersecting probability} (denoted by $p$) of any two requests (say $R_1$ and $R_2$), generated from $\mathcal{D}$ and routed by $(p_1,p_2,...,p_K)$, is the probability that their lightpaths $P_1$ and $P_2$ share at least one common fiber link, \textit{i.e.}, $P_1 \cap P_2 \neq \emptyset$.
\end{definition}

Given $\mathcal{R}=\{R_1,R_2,...,R_n\}$, the vertex set of the conflict graph $\hat{G}(\hat{V},\hat{E})$ is that $\hat{V}=\{\hat{v}_1,\hat{v}_2,...,\hat{v}_n\}$. Thus, the edge set $\hat{E}$ determines the final $\chi(\hat{G})$. For any two vertices, say $\hat{v}_1, \hat{v}_2 \in \hat{V}$, the probability that they are adjacent in the conflict graph $\hat{G}$ is equal to the intersecting probability $p$. We introduce an important Lemma in random graph theory which reveals the relation between the intersecting probability $p$ and the chromatic number $\chi(\hat{G})$.

\begin{lemma}\cite{s3:b5}
\label{lem:kp}
Let $|\hat{V}|=n$, $p$ be the intersecting probability, and $\hat{G}$ be the conflict graph, then we have $\chi(\hat{G})=(\frac{1}{2}+o(1))\cdot \log(\frac{1}{1-p})\cdot \frac{n}{\log n}$, where $o(1)$ is an infinitesimal of $n$. 
\end{lemma}

From Lemma \ref{lem:kp}, we can see that there is a strongly positive correlation between the $\chi(\hat{G})$ and the intersection probability $p$. More specifically, the smaller the intersection probability $p$, the smaller the $\chi(\hat{G})$, which conforms to our intuition. Meanwhile, by Theorem \ref{the:chropt}, the smaller the chromatic number $\chi(\hat{G})$, the smaller the $|opt(\hat{G})|$. Therefore, we obtained the first important theoretical chain in Fig. \ref{fig:chian the}. 

%Consequently, \textbf{the optimal routing decision should be the one leading to the minimum intersecting probability $p$.}  

\begin{figure}[!htb]
\centering
 \includegraphics[height=0.075\textwidth]{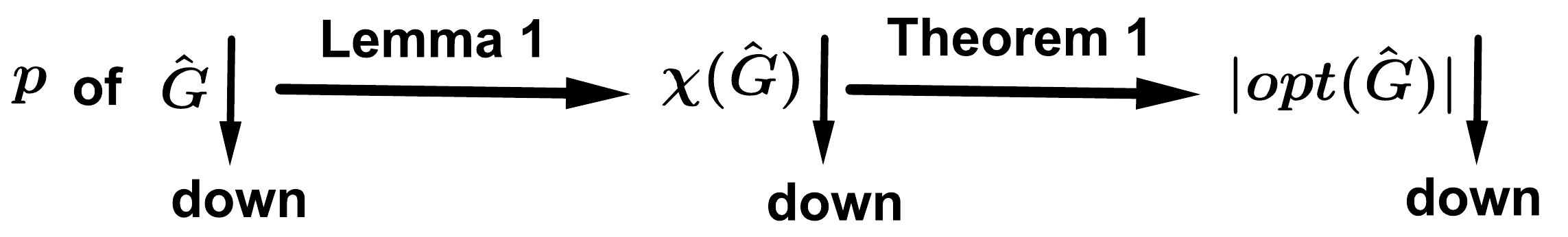}
\caption{Theoretical Chain 1}
\label{fig:chian the}
\end{figure}

Now, the question is that how to determine the intersecting probability $p$. Obviously, the $p$ depends on the three factors: \circled{1} the network topology $G(V,E)$, \circled{2} the traffic distribution $\mathcal{D}$ and \circled{3} the routing scheme $(p_1,p_2,...,p_K)$. We use a simple example below to illustrate that.

\begin{figure}[!htb]
\centering
 \begin{subfigure}[Network Configuration.]{
        \label{subfig:ncqs}
		\includegraphics[height=0.2\textwidth]{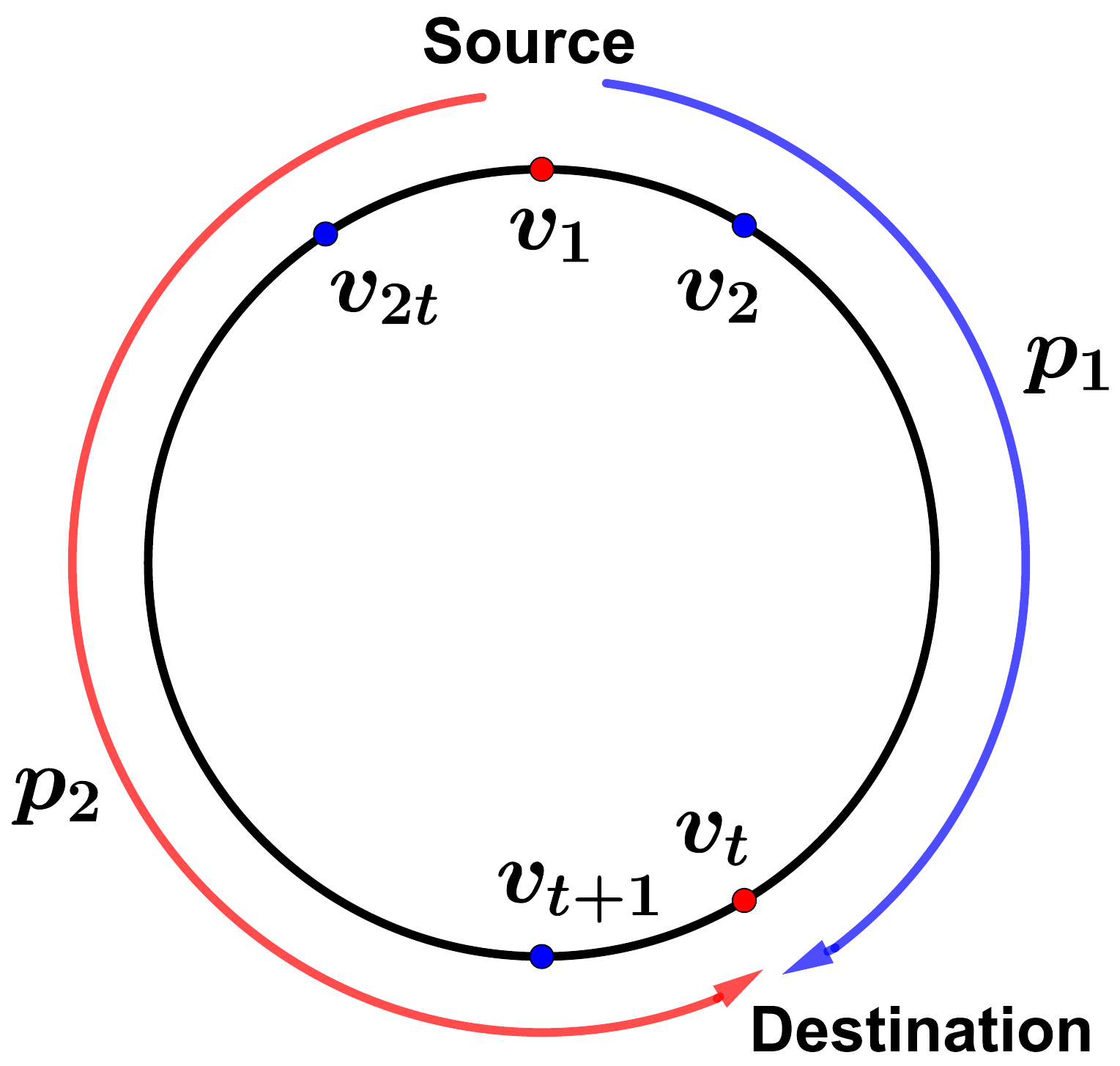}}
    \end{subfigure}
      \begin{subfigure}[Intersecting Probabilities.]{
        \label{subfig:tcgreip}
		\includegraphics[height=0.25\textwidth]{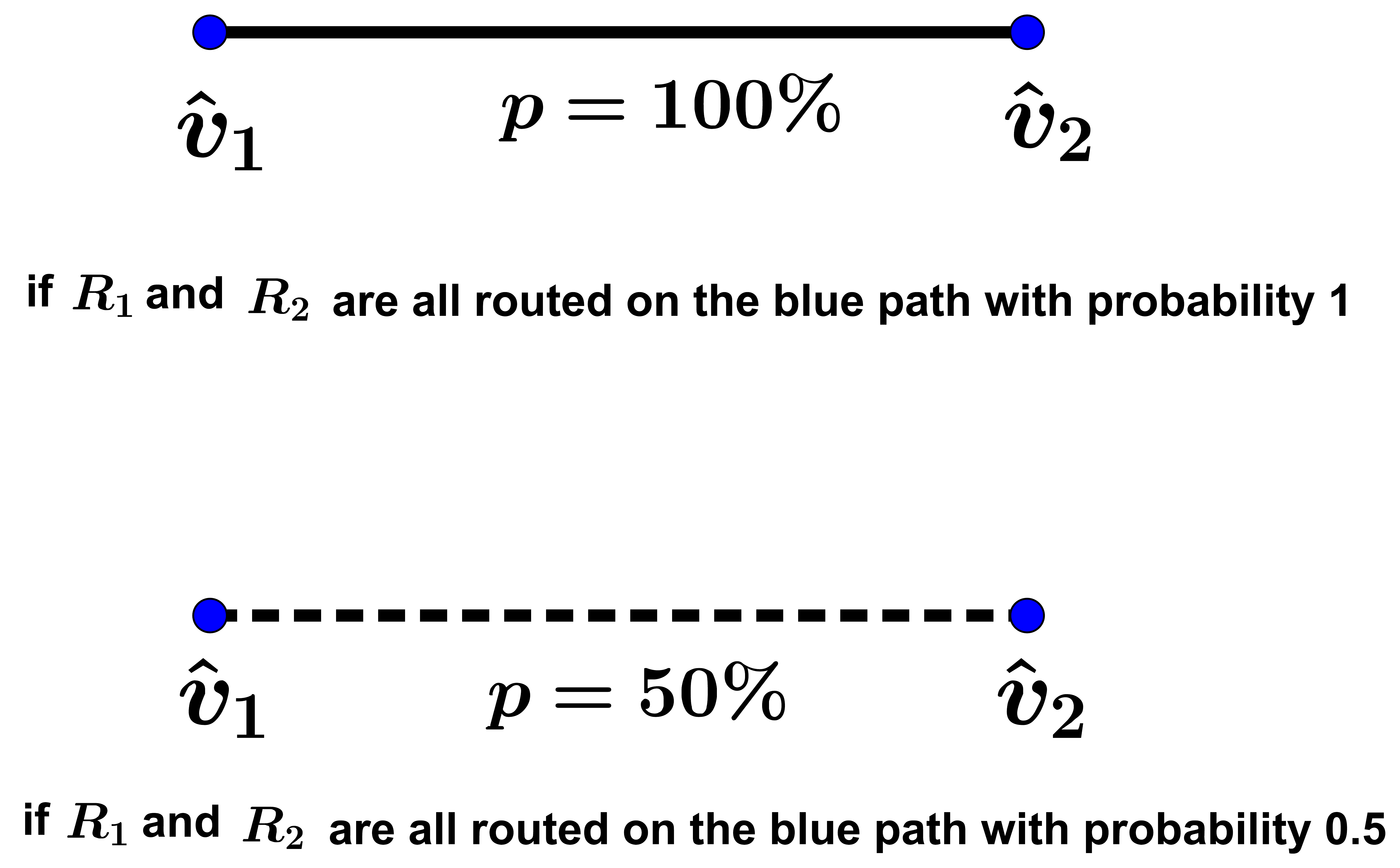}}
    \end{subfigure}
\caption{An example of intersecting probability.}
\label{fig:eaxmp}
\end{figure}

\begin{exmp}
\label{exam1}
The network topology of the EON is a bidirectional cycle $G(V,E)=[v_1v_2...v_{2t}v_1]$, whose vertice are labeled clockwise from $v_1$ to $v_{2t}$ as shown in Fig. \ref{fig:eaxmp}. We assume the following traffic distribution $\mathcal{D}$: The occurrence probability of the source-destination pair $(v_1,v_t)$ is $w_{v_1v_t}=100\%$ while it is zero for the others. We can see there are two candidate paths marked on blue and red respectively.

If the routing scheme is $(1,0)$, \textit{i.e.}, it routes all the requests on the first shortest path (the blue one), then the intersecting probability $p$ is $100\%$, which yields a complete conflict graph $\hat{G}$ with the biggest chromatic number, thereby resulting in the maximum $|opt(\hat{G})|$. Obviously, it is a bad routing scheme. An intuitively optimal routing scheme should be $(0.5,0.5)$, \textit{i.e.}, half of the requests are routed on the red one and the remaining half on the blue one, and the intersecting probability becomes $p=50\%$.

% \begin{figure}[!htb]
% \centering
%  \includegraphics[height=0.25\textwidth]{ZXX.pdf}
% \caption{An example of intersecting probability.}
% \label{fig:eaxmp}
% \end{figure}
\end{exmp}

From the above example, it is not difficult to observe that the routing scheme should take into account the network topology and traffic distribution. Besides, we can also find that there is a lower bound on the intersecting probability $p$. In other words, no matter what routing scheme is used, there is a minimum intersecting probability $p_{min}$ that we can not decrease. For instance $p_{min}=50\%$ for the above example. The limitation on $p_{min}$ obviously roots in the network topology $G$ and traffic distribution $\mathcal{D}$.

\subsection{Conflict Coefficients}
\label{sub:the conflict coefficient}
%Give two requests $R_i$ and $R_j$ on $G$, both generated by $\mathcal{D}$, 
%Before discussing how traffic distribution $\mathcal{D}$ and network topology of $G$ impact the intersecting probability $p$, 
Here, we propose a new concept named \textbf{conflict coefficient} of an EON, which plays an important role for our next theoretical analysis.

\begin{definition}
The \textit{conflict coefficient} $\theta_{ij}$ is \textbf{the intersecting probability} of any two requests (say $R_1$ and $R_2$) generated from $\mathcal{D}$, \textbf{under the condition that} $R_1$ and $R_2$ are respectively routed on the $i$-th and $j$-th shortest paths of their own source-destination pairs in the EON $G$.
\end{definition}

Thus, the $\theta_{ij}$ is only related to network topology $G(V,E)$ and traffic distribution $\mathcal{D}$. The $\theta_{ij}$ can be computed as follows. We first generate the $i$-th and $j$-th shortest paths for all $|V|\times (|V|-1)$ source-destination pairs in $G(V,E)$. We construct a matrix $M^{\theta_{ij}}$ in Table \ref{tab:proslij}.

\begin{table} [!htbp]
\centering
\setlength{\extrarowheight}{1mm}
\caption{Matrix $M^{\theta_{ij}}$ of the conflict coefficient $\theta_{ij}$.}
\label{tab:proslij}
%\begin{tabular}{ l | p{1.3cm}  p{1.3cm}   p{1.3cm} }
\noindent\begin{tabular}{c@{}l}
\begin{tabular}{cc|ccccc}
& & \multicolumn{5}{c}{On the $j$-th path}  \\
 & & .  & . & $(s_2,d_2)$ $w_{s_2d_2}$ & . & .\\
\hline
 & . & . & . & . & . & .\\
On &. & . & . & . & . & .\\
the& $(s_1,d_1)$ & . & . & $w_{s_1d_1} \times w_{s_2d_2}$ & . & .\\ 
$i$-th &$w_{s_1d_1}$ & . & . & . & . & .\\
path & . & . &  . & . & . & . \\
 &. & . &  . & . & . & . \\
 &   \multicolumn{1}{c}{} & \multicolumn{5}{@{}l@{}}{%
      \raisebox{.5\normalbaselineskip}{%
      \rlap{$\underbrace{\hphantom{\mbox{\hspace*{\dimexpr23\arraycolsep}}}}_{|V| \times (|V|-1)}$}}
    }
\end{tabular}
 &
$\begin{array}{l}
    \MyLBrace{9ex}{\tiny{$|V| \times (|V|-1)$}}
 \end{array}$
\end{tabular}
\end{table}
% There are some differences between TABLE \ref{tab:proslij} and above tables.
The top row in Table \ref{tab:proslij} represents all the $|V|\times (|V|-1)$ source-destination pairs, which are routed on their own $j$-th shortest paths, while the leftmost column represents the $|V|\times (|V|-1)$ source-destination pairs, which are routed on their own $i$-th shortest paths. The weight $w_{sd}$ for each source-destination pair $(s,d)$ is its occurrence probability determined by the traffic distribution $\mathcal{D}$. If the $i$-th shortest path of $(s_1,d_1)$ intersects with the $j$-th shortest path of $(s_2,d_2)$ in $G(V,E)$, then this entry is $w_{s_1d_1} \times w_{s_2d_2}$, otherwise 0. Then, we have Theorem \ref{the:conflictM1M2}.

\begin{theorem}
\label{the:conflictM1M2}
The conflict coefficient $\theta_{ij}=\sum M^{\theta_{ij}}$, where $\sum M^{\theta_{ij}}$ represents the sum of all entries in the matrix.
\end{theorem}

\begin{proof}
According to the definition, $\theta_{ij}$ is the intersecting probability of any two requests, say $R_1$ routed on the $i$-th shortest path and $R_2$ on the $j$-th shortest path. Since both $R_1$ and $R_2$ are generated from the traffic distribution $\mathcal{D}$ in $G(V,E)$, the probability that the source-destination pairs of $R_1$ and $R_2$ being $(s_1,d_1)$ and $(s_2,d_2)$ respectively is $w_{s_1d_1} \times w_{s_2d_2}$. Therefore, $\theta_{ij}=\sum_{\big((s_1,d_1)^i \cap (s_2,d_2)^j\neq \emptyset\big)} w_{s_1d_1} \times w_{s_2d_2} $, where $(s_1,d_1)^i \cap (s_2,d_2)^j\neq \emptyset$ means that the $i$-th shortest path of $(s_1,d_1)$ intersects with the $j$-th shortest path of $(s_2,d_2)$ in the $G(V,E)$. Finally, it is easy to see that $\sum_{\big((s_1,d_1)^i \cap (s_2,d_2)^j\neq \emptyset\big)} w_{s_1d_1} \times w_{s_2d_2}=\sum M^{\theta_{ij}}$ and the proof follows.
\end{proof}

%The conflict coefficient $\theta_{ij}$ is an integral reflection on the impact of both the traffic distribution $\mathcal{D}$ and the EON topology $G$. 
Following the same idea, we can compute all the conflict coefficients of an EON $G$ under a traffic distribution $\mathcal{D}$, which compose \textbf{a real symmetric Conflict Matrix (CM)} as below. 

%CM is an important evaluation metric for an EON $G$ under a specific traffic distribution $\mathcal{D}$ as we will see later.

\[
CM=\begin{bmatrix}
    \theta_{11}       & \theta_{12} & \theta_{13} & \dots & \theta_{1K} \\
    \theta_{21}       & \theta_{22} & \theta_{23} & \dots & \theta_{2K} \\
    \theta_{K1}       & \theta_{K2} & \theta_{K3} & \dots & \theta_{KK}  \\ 
\end{bmatrix}
\]

% Conflict coefficients $\theta$ can also be extended to WDM networks and multicast traffic, where each entry in the matrix $M^\theta_{1,2}$ represents a multicast tree rather than a lightpath. 

\subsection{Global Optimal Formulation (GOF)}

Under the routing scheme $(p_1,p_2,...,p_K)$, for any two requests, say $R_1$ and $R_2$, the probability that the $R_1$ are routed on the $i$-th shortest candidate path and $R_2$ on the $j$-th shortest candidate path is $p_ip_j$. Combing with the conflict coefficient $\theta_{ij}$, the conditional intersecting probability that any two requests, routed on $i$-th and $j$-th shortest candidate paths respectively, intersect is $\theta_{ij}p_ip_j$. Thus, the intersecting probability $p$ is equal to the sum of all the conditional intersecting intersecting probabilities. We obtain a Global Optimal Formulation (GOF) as follows, which is a quadratic programming and determines the intersecting probability $p$:
\begin{equation}
\label{eqn: objet}
\begin{aligned}
\quad p= \sum_{1 \leq i,j \leq K}\theta_{ij}p_ip_j \textrm{~~~~(\textbf{GOF})},\\
\end{aligned}
\end{equation}
\begin{align}
\label{eqn:prs}
\textbf{s.t.} \quad & \sum^K_{i=1} p_i=1, \\
\label{eqn: positive}
& p_i \geq 0, \quad 1 \leq i \leq K.
\end{align}
Here, $K$ is the number of pre-computed candidate paths for each source-destination pair, and $\theta_{ij}, \forall i,j$ are the conflict coefficients. In GOF, $K$ is predetermined, and $\theta_{ij}$ are the parameters determined by the  the network topology $G$ and traffic distribution $\mathcal{D}$, while $p_i$ is determined by the routing scheme. Thus, the complexity of the quadratic programming is with one constraint, $K$ variables and $K^2$ parameters. By now, we summarize what we have obtained in Fig. \ref{fig:chian the2} which is another important theoretic chain.

\begin{figure}[!htb]
\centering
 \includegraphics[height=0.145\textwidth]{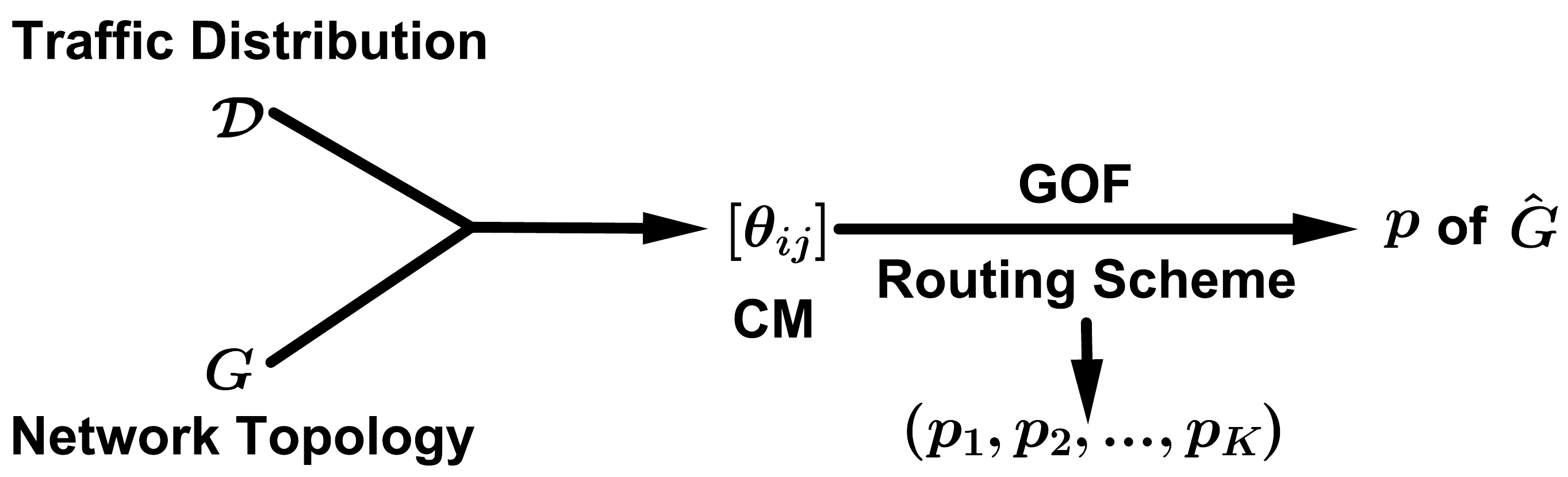}
\caption{Theoretical Chain 2}
\label{fig:chian the2}
\end{figure}

\textbf{Combing Theoretic Chain 1 in Fig. \ref{fig:chian the} with Theoretic Chain 2 in Fig. \ref{fig:chian the2}, we finally figured out how the network topology $G$, traffic distribution $\mathcal{D}$ and routing scheme $(p_1,p_2,...,p_K)$ impact on the optimal spectrum usage.} 

More specifically, the CM, composed of the conflict coefficients, embodies the impact of the network topology $G$ and traffic distribution $\mathcal{D}$. The CM can be viewed as the capacity of an EON $G$ under a specific traffic distribution $\mathcal{D}$. Besides, we can also see that the routing scheme, which determines the array $(p_1,p_2,...,p_K)$, is another important factor to determine the intersecting probability $p$. A good routing scheme should get a proper array $(p_1,p_2,...,p_K)$ resulting in a small intersecting probability $p$. In fact, after obtaining the CM, we can minimize the GOF to obtain the optimal $(p_1,p_2,...,p_K)$  which reaches the minimum intersecting probability $p_{min}$. We also use Example \ref{exam1} to illustrate this.

In the Example \ref{exam1} of Fig. \ref{fig:eaxmp}, we set $K=2$. According to  the topology of cycle-EON and traffic distribution ($w_{v_1v_t}=100\%$), the conflict coefficients are $\theta_{11}=1$, $\theta_{12}=\theta_{21}=0$ and $\theta_{22}=1$. Therefore, we can compute the intersecting probability $p$ of any two requests in this example as follows.
\begin{equation}
\label{eqn: tcpsssssss}
\begin{aligned}
& \quad \quad ~ p=p^2_1+p^2_2,\\
&s.t. \quad p_1+p_2=1.
\end{aligned}
\end{equation}
Obviously, the intersecting probability $p$ in Eq. (\ref{eqn: tcpsssssss}) reaches the minimum when $p_1=p_2=0.5$. It means the optimal routing scheme should be $(0.5,0.5)$, which conforms to the intuition.

% Therefore, the intersecting probability $p$ of any two requests is determined by Eq. \ref{eqn: objet}. \textbf{Thus, the optimal routing decision is determined by the optimal decision on $p_i$.} 

% Therefore, the optimal routing decision is to assign an appropriate proportions among all $p_{i, 1\leq i \leq K}$ so that the final intersecting probability $p$ is minimum.
%  For any two requests, say $R_1$ and $R_2$, under the condition that $R_1$ routed on the $i$-th candidate path and $R_2$ on the $j$-th candidate path, the intersecting probability is $\theta_{ij}$. 

% After the optimal $(p_1,p_2,...,p_K)$ is obtained by solving GOF, to serve a request, we can assign to it the $i$-th candidate path according to the probability $p_i$. In such way, the intersecting probability $p$ of the conflict graph is minimum. Then, we can assign FS set to each request by using spectrum assignment algorithms.

% Above fact implies that we need to comprehensively take the request distribution and the topology of EONs into account to build a global optimal formulation to reduce the intersecting probability.

%The smaller the conflict coefficients, the smaller the minimum intersecting probability of the conflict graph $\hat{G}$. According to Theorem \ref{the:chropt}, the optimal MUFI $opt(\hat{G})$ is always limited by $\chi(\hat{G})$, while the minimum $\chi(\hat{G})$ is determined by the conflict coefficients as demonstrated in Lemma \ref{lem:kp} and Eq. (\ref{eqn: objet}). This is the mechanism forming the performance ceiling of an EON under a specific traffic distribution.

\subsection{Further Discussions}
%\subsubsection{For the $K$}
%From GOF, we can obviously see that with the growth of $K$, the entire feasible region of GOF is expanded that may help to decrease the minimum intersecting probability $p$.
%Naturally, a bigger $K$ implies more computation time for solving GOF. But, notice that it is a disposable consumption, \textit{i.e.}, when the EON $G(V,E)$ and traffic distribution $\mathcal{D}$ are given, we just need to solve the GOF once beforehand. Therefore, it is worth using a large $K$ to determine the minimum intersecting probability in EONs. 

%In fact, our proposed GOF is not just restricted to unicast traffic in EONs. It can also be extended for RWA in WDM networks and multicast traffic, which will be discussed below briefly. 

\subsubsection{For the G and $\mathcal{D}$}
There is always a lower bound for the intersecting probability $p$ in the GOF, \textit{i.e.}, the minimum intersecting probability $p_{min}$, that we can not decrease, no matter how to optimize the routing scheme. This lower bound roots in the network topology and traffic distribution by these conflict coefficients. In consequence, this limitation on $p_{min}$ restricts the final performance of the RSA. How to optimize the network topology and traffic distribution to adjust these conflict coefficients to lower the $p_{min}$ is an important direction in future works. 

\subsubsection{For the Routing Scheme}
%Although we only focus on unicast communications in this paper, 
In this paper, we simplify the representation of the routing scheme, which is represented by the probability array of $(p_1,p_2,...,p_K)$. But, in reality, the array of $(p_1,p_2,...,p_K)$ of many routing scheme is hard to estimate in advance. After we obtain the optimal array $(p_1,p_2,...,p_K)$ by minimizing the GOF, how to utilize this optimal percentage array on each candidate path to plan a better routing scheme is another topic that can go in-depth in future works. In this paper, we deal with this matter as following: After a requests $R(s,d)$ is generated, we take the probability $p_i$ to select the $i$-th candidate path of $(s,d)$ as the routed lightpath of $R$.

%In multi-cast situation, each source-multi-destinations $(s,D)$, where $D$ is the set of destination nodes $d \in D$, has its own occurrence probability. For any two multicast requests $(s_1,D_1)$ and $(s_2,D_2)$, we may similarly calculate the corresponding conflict coefficients on different candidate trees instead of candidate paths, which, of course, is more complicated than in the unicast and needs more efforts to work out. 

%First, we calculate some candidate multicast trees for each possible request. Then following the similar way to solve the GOF, the optimal proportions for each candidate tree can be obtained.

% \subsection{Intersecting Probability}
% In general case, there is no regular way to estimate the chromatic number of the conflict graph over EONs. Hence, we also need the intersecting probability to approach the chromatic number. The smaller the intersecting probability is, the smaller the chromatic number is which leads to a better optimal solution as shown in Theorem \ref{the:core}. This fact conforms to the intuition. Given an EON and a distribution for the requests, how to reduce the intersecting probability is the key point of designing RSA algorithms.

\begin{figure*}[!htb]
\centering
 \begin{subfigure}[Ring (12 nodes and 12 bidirectional fiber links)]{
        \label{subfig:reon}
		\includegraphics[height=0.27\textwidth]{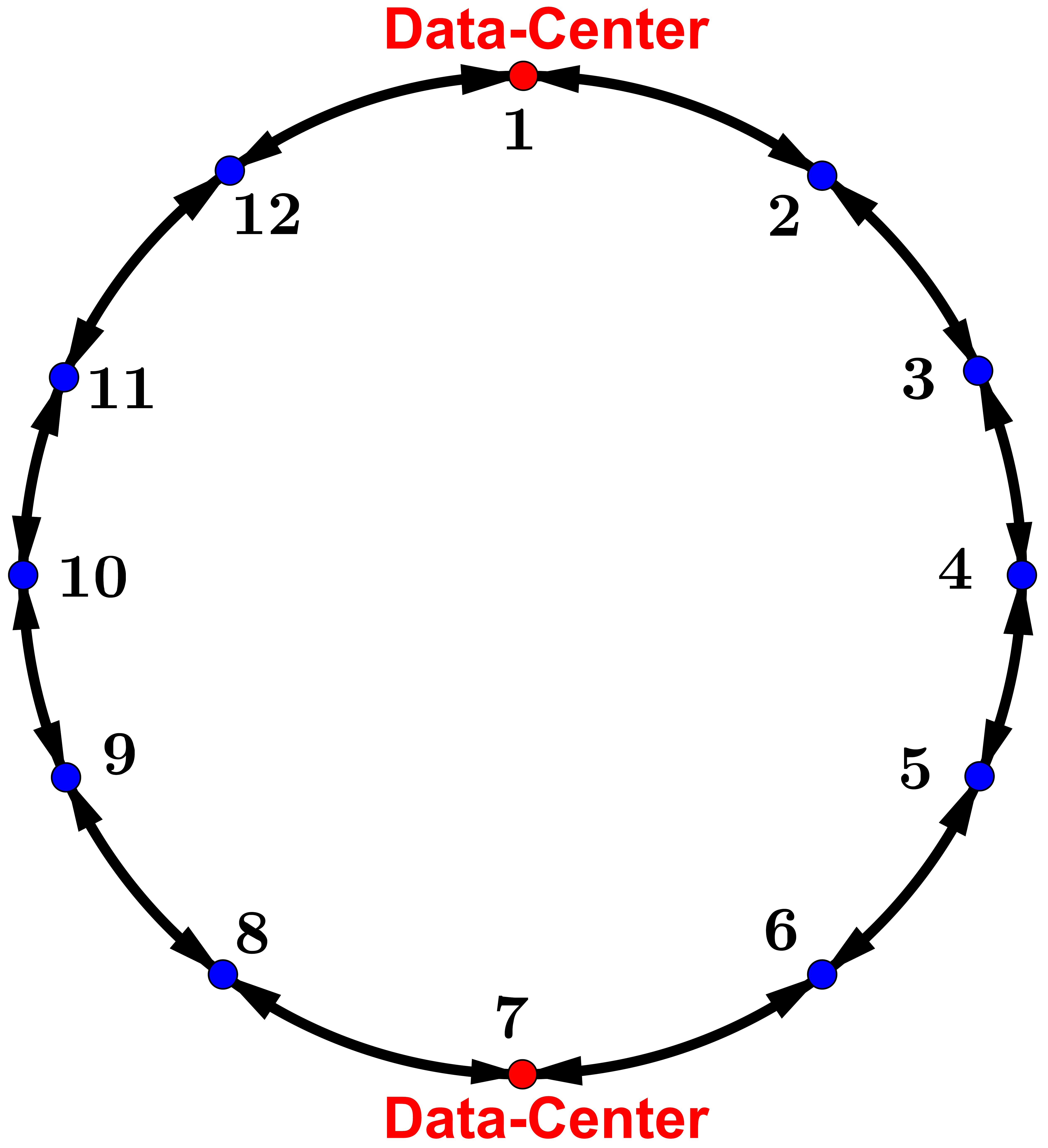}}
 \end{subfigure}
 \begin{subfigure}[NSFNET (14 nodes and 22 bidirectional fiber links)]{
        \label{subfig:nseon}
		\includegraphics[height=0.2\textwidth]{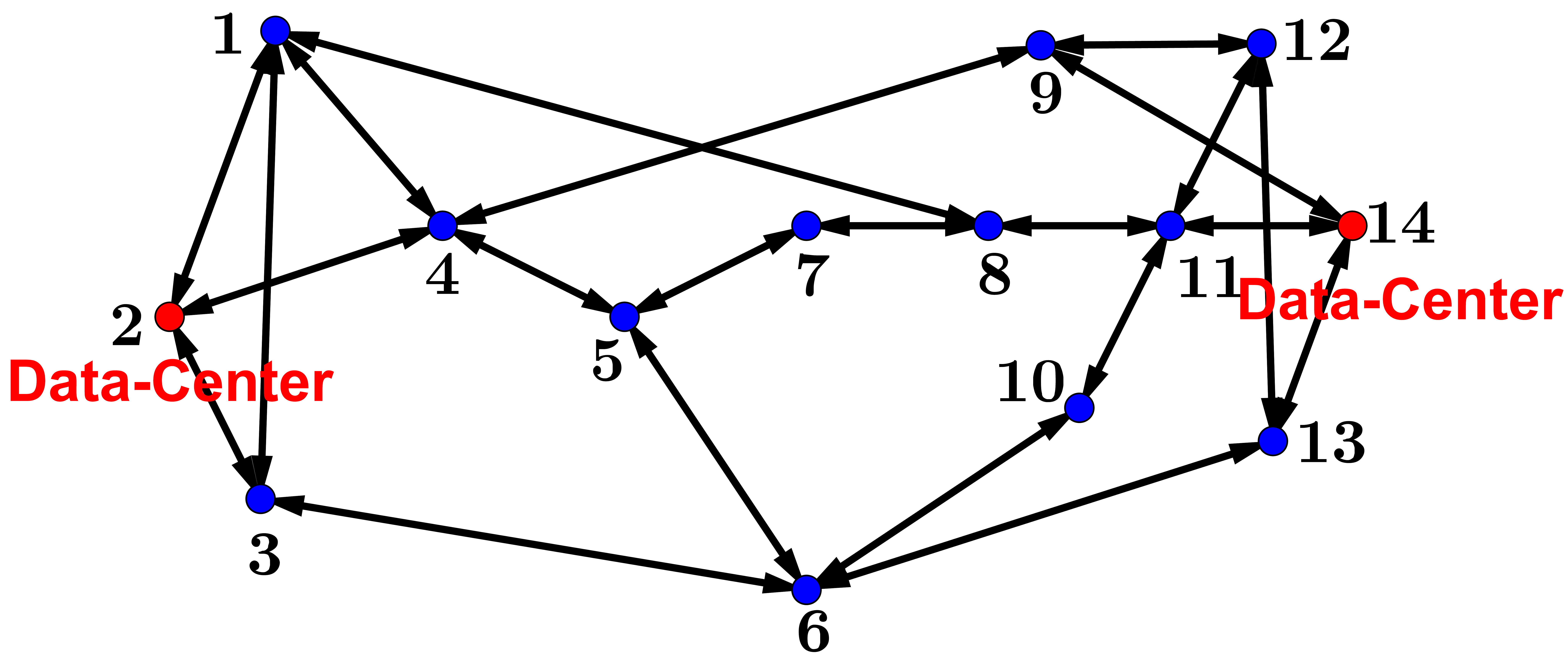}}
 \end{subfigure}
 \begin{subfigure}[NJ-LATA (11 nodes and 23 bidirectional fiber links)]{
        \label{subfig:njeon}
		\includegraphics[height=0.25\textwidth]{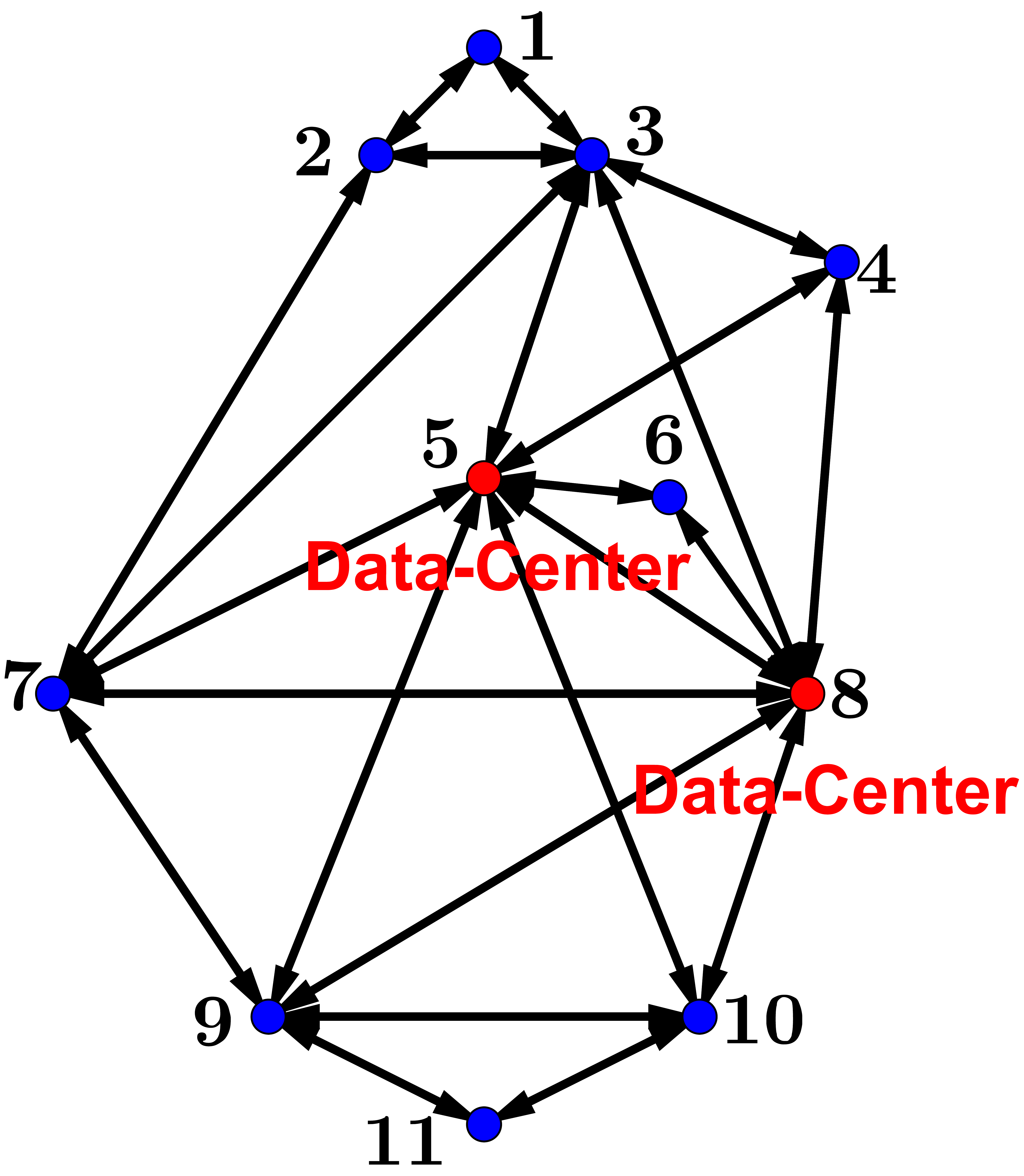}}
 \end{subfigure}
\caption{Three realistic EON topologies.}
\label{threon}
\end{figure*}

\section{CM Estimation and Optimal Routing Scheme in Realistic EONs}
\label{sec:ciceon}
In this section, we estimate the conflict coefficients and solve the corresponding GOFs in three realistic EONs under two traffic distributions respectively. The three EONs are, as shown in Fig. \ref{threon}, the Ring, NSFNET and NJ-LATA \cite{ss4:b5} respectively. The three EONs are of almost the same size in terms of the number of nodes. Thus, we shall also compare the three EONs from the perspective of intersecting probability. Since in the Ring, there are only two candidate paths for each source-destination pair, to make a fair comparison, we set $K=2$ for all EONs by pre-computing the first and second shortest paths for each source-destination pair. Therefore, the intersecting probability $p$ in GOF can be written as

\[
p =
  \begin{bmatrix}p_1 & p_2\end{bmatrix}
  \begin{bmatrix} \theta_{11} & \theta_{12} \\ \theta_{21} & \theta_{22} \end{bmatrix}
  \begin{bmatrix} p_1 \\ p_2 \end{bmatrix}
\]
where, $p_1+p_2=1$, or 

\begin{equation}
\label{equ:corep}
p=\theta_{11}\times p^2_1+2\theta_{12}\times p_1p_2+\theta_{22}\times p^2_2.
\end{equation}

In the following, we compute $p$ under two traffic distribution scenarios respectively: uniform and weighted. 

\subsection{Uniform Traffic Distribution}
\label{subsec:cd56}

In this subsection, we consider a uniform traffic distribution $\mathcal{D}$ in the three EONs, \textit{i.e.}, each source-destination pair occurs with the same occurrence probability $\frac{1}{|V|\times (|V|-1)}$. Following the computing method in Section \ref{sub:the conflict coefficient}, we can get the corresponding conflict coefficients and GOFs of the three EONs under uniform distribution.

\subsubsection{The Ring under the Uniform Traffic Distribution}
\hfill

The CM in this case is

\[
CM^{Uni}_{Ring}=\begin{bmatrix}
    0.2328       & 0.4360  \\
    0.4360       & 0.5014 \\
\end{bmatrix}
\]

and we can thus get the intersecting probability via the GOF

\begin{equation}
\label{equ:ringuni}
p=0.2328\times p^2_1+0.8720\times p_1p_2+0.5014\times p^2_2.
\end{equation}

where,  $p_1+p_2=1$.

By minimizing Eq. (\ref{equ:ringuni}), we get the minimum intersecting probability $p_{min}=23.28\%$. This can be achieved by the optimal routing scheme: $p_1=1$ and $p_2=0$, \textit{i.e.}, $(1,0)$.

\subsubsection{The NSFNET under the Uniform Distribution}
\hfill

The CM in this case is that

\[
CM^{Uni}_{NSF}=\begin{bmatrix}
    0.0979       & 0.1377  \\
    0.1377       & 0.2042 \\
\end{bmatrix}
\]

and thus the intersecting probability by the GOF is

\begin{equation}
\label{equ:nsfuni}
p=0.0979\times p^2_1+0.2754\times p_1p_2+0.2042\times p^2_2.
\end{equation}

where,  $p_1+p_2=1$.

By minimizing Eq. (\ref{equ:nsfuni}), the minimum intersecting probability $p_{min}=9.79\%$ and the optimal routing scheme is $p_1=1$ and $p_2=0$, \textit{i.e.}, $(1,0)$.

\subsubsection{The NJ-LATA under the Uniform Distribution}
\hfill

The CM in this case is that

\[
CM^{Uni}_{NJ}=\begin{bmatrix}
    0.0901       & 0.0852  \\
    0.0852       & 0.1157 \\
\end{bmatrix}
\]

and thus the intersecting probability by the GOF is

\begin{equation}
\label{equ:njuni}
p=0.0901\times p^2_1+0.1704\times p_1p_2+0.1157\times p^2_2.
\end{equation}

where,  $p_1+p_2=1$.

By minimizing Eq. (\ref{equ:njuni}), we have the minimum intersecting probability $p_{min}=8.94\%$ and the optimal routing scheme: $p_1=0.8621$, $p_2=0.1379$, \textit{i.e.}, $(0.8621, 0.1379)$.

\subsection{Weighted Traffic Distribution}
Nowadays, EONs begin to support new networking capabilities and demanding network
services such as data centers and cloud. Hence,  the Optical Cross-Connect (OXC) in EONs connected to data-centers will produce a large amount of traffic among them, for instance data migration and content provisioning. This kind of traffic contributes to the majority of the total traffic. In other words, the nodes connected to data centers tend to have a much higher occurrence probability to serve as the source or destination of a request than the other nodes \cite{ss4:b3,ss4:b4}.

In this subsection, we assume that there are two data-center nodes in the three considered EONs for simplicity. Both the two data-center nodes have the same big  occurrence probability ($45\%$ in this paper) to be involved in a request (as source or destination), while the other EON nodes equally share the remaining $10\%$ possibility. It should be noted the value of the occurrence probability given here is just for illustrative purpose, which can be in fact arbitrary. We call this distribution as the weighted traffic distribution in this paper.

\subsubsection{The Ring under the Weighted Traffic Distribution}
\hfill

We assume that the nodes $1$ and $7$ in the Ring are the two data-center nodes.  Thus, we have its CM

\[
CM^{Con}_{Ring}=\begin{bmatrix}
    0.3829       & 0.1766  \\
    0.1766       & 0.5000 \\
\end{bmatrix}
\]

and its intersecting probability

\begin{equation}
\label{equ:ringcon}
p=0.3829\times p^2_1+0.3532\times p_1p_2+0.5000\times p^2_2.
\end{equation}

where,  $p_1+p_2=1$.

By minimizing Eq. (\ref{equ:ringcon}), we obtain the minimum intersecting probability $p_{min}=30.26\%$ and the optimal routing scheme: $p_1=0.6105$ and $p_2=0.3895$, \textit{i.e.}, $(0.6105, 0.3895)$. 

\subsubsection{The NSFNET with Weighted Traffic Distribution}
\hfill

We assume that the nodes $2$ and $14$ in NSFNET are the two data-center nodes. Then, we have its CM     
\[
CM^{Con}_{NSF}=\begin{bmatrix}
    0.3554       & 0.2119  \\
    0.2119       & 0.3982 \\
\end{bmatrix}
\]

and its intersecting probability

\begin{equation}
\label{equ:nsfcon}
p=0.3554\times p^2_1+0.4238\times p_1p_2+0.3982\times p^2_2.
\end{equation}

where,  $p_1+p_2=1$.

By minimizing Eq. (\ref{equ:nsfcon}), we get the minimum intersecting probability $p_{min}=29.30\%$. The optimal routing scheme is $p_1=0.5648$ and $p_2=0.4352$, \textit{i.e.}, $(0.5648, 0.4352)$.

\subsubsection{The NJ-LATA under the Weighted Traffic Distribution}
\hfill

Nodes $5$ and $8$ are assumed to the two data-center nodes in NJ-LATA. In this case, the CM becomes

\[
CM^{Con}_{NJ}=\begin{bmatrix}
    0.2758       & 0.0616  \\
    0.0616       & 0.3306 \\
\end{bmatrix}
\]

and the intersecting probability is given by

\begin{equation}
\label{equ:njcon}
p=0.2758\times p^2_1+0.1232\times p_1p_2+0.3306\times p^2_2.
\end{equation}

where, $p_1+p_2=1$.

Thus, the minimum intersecting probability is $p_{min}=18.08\%$ and the optimal routing scheme is: $p_1=0.5568$ and $p_2=0.4432$, \textit{i.e.}, $(0.5568, 0.4432)$.

Now, we compare the minimum intersecting probability of the three EONs under the two traffic distributions in Table \ref{tab: tmipiss}.

\begin{table}[!htb]
\centering
  \caption{Comparison of the Minimum Intersecting Probability}
\label{tab: tmipiss}
\resizebox{0.47\textwidth}{!}{\begin{tabular}{|M{1.9cm}|M{1.2cm}|M{1.2cm}|M{1.3cm}|}
\hline
 \diagbox{Traffic}{EON} & Ring & NSFNET & NJ-LATA \\
\hline
Uniform Distribution &  23.28\%  & 9.79\% & 8.94\% \\
\hline
Weighted Distribution& 30.26\% & 29.30\% & 18.08\% \\
\hline
\end{tabular}}
\end{table}

The minimum intersecting probabilities of the Ring and NSFNET under weighted distribution are the two highest ones, 30.26\% and 29.30\% respectively, while that of the NJ-LATA and NSFNET under uniform distribution are the two lowest ones, 8.94\% and 9.79\% respectively. 

Thus, if taking their own optimal routing schemes and with a same spectrum assignment method, the final MUFIs of the Ring and NSFNET under the weighted distribution should be the two maximums, and that of NJ-LATA and NSFNET under uniform distribution should be the two minimums. 

%As for the two scenarios of the Ring and NSFNET under the weighted distribution (the same for the NJ-LATA and NSFNET under uniform distribution),  the difference of their minimum intersecting probabilities is only within about 1\%. As we shall see in the simulations, when the difference is less than 3\%, the effectiveness of the intersecting probability is not significant. Since the conflict graphs are so close from the viewpoint of intersecting probability that the final MUFI  depends more on the EON topology, which needs more delicate works to figure out.

\section{Numerical Results}
\label{sec:numres}

%As discussed above, we derive the impact of traffic distribution $\mathcal{D}$ and EON topology $G$ on the lightpath routing by combing Theoretic Chains 1 and 2 in Figs. \ref{fig:chian the} and \ref{fig:chian the2} respectively.
In this section, we verify the effectivenesses of the two theoretic chains by simulations:

 \begin{itemize}
 \item The effectiveness of Theoretic Chain 2 in Fig. \ref{fig:chian the2}, \textit{i.e.}, the conflict coefficients and the computing method for intersecting probability. In other words, whether the theoretical intersecting probability computed by the GOF can fit in with the realistic one.
 \item The effectiveness of Theoretic Chain 1 in Fig. \ref{fig:chian the}, \textit{i.e.}, the intersecting probability itself. In other words, whether the intersecting probability is positively correlated to the final MUFI of the RSA.
 \end{itemize}

We conduct simulations on the three EONs under both the uniform and the weighted traffic distributions described before. In our simulations, after the lightpath routing, we utilize a same spectrum assignment algorithm MRSA in \cite{s1:b2} to assign FS sets to requests. We consider six scenarios in our simulations, which are labeled in Table \ref{tab: ssos}.

\begin{table}[!htb]
\centering
  \caption{Six Simulation Scenarios}
\label{tab: ssos}
\resizebox{0.47\textwidth}{!}{\begin{tabular}{|M{1.9cm}|M{1.2cm}|M{1.2cm}|M{1.3cm}|}
\hline
 \diagbox{Traffic}{EON} & Ring & NSFNET & NJ-LATA \\
\hline
Uniform Distribution &  R-U  & NSF-U & NJ-U \\
\hline
Weighted Distribution& R-W & NSF-W & NJ-W \\
\hline
\end{tabular}}
\end{table}

In each scenario, we first pre-compute two candidate paths for each source-destination pair, \textit{i.e.}, the first and second shortest paths as mentioned in Section \ref{sec:ciceon}. Eleven routing schemes will be conducted and compared in each scenario by increasing $p_1$ (the percentage of requests routed on the first shortest path) from 0 to 1 with a step of 0.1. Hence, there are 66 cases  in total for the six scenarios. Meanwhile, for each case, we compare the realistic intersecting probability $p$ with the one theoretically estimated by the GOF through substituting the value of $p_1$ into the corresponding formulation obtained in Section \ref{sec:ciceon}. Here the realistic intersecting probability is $\dfrac{\#\{e\}}{\binom{|\mathcal{R}|}{2}}$, where $\#\{e\}$ is the  number of realistic edges in the conflict graph.

Besides, from viewpoint of intersecting probability, we shall also uniformly  analyze the six scenarios and 66 cases as follows.

 \begin{itemize}
 \item  The final RSA performance of the six scenarios under their own optimal routing schemes.
 \item The performance differences in the 66 cases. 
 \end{itemize}

For the bandwidth range and guard band size, we set $\alpha=1$, $\beta=4$, and $GB=1$ respectively. 
%Besides, due to the fact that RWA is actually a special case of RSA by setting $\alpha=\beta=1$ and $GB=0$, we also take into account the plain RWA case in our simulations.
The number of requests is set as $|\mathcal{R}|=1000$ in each simulation. We repeat each simulation 100 times under the same circumstance to ensure sufficient statistical accuracy, and a 95\% confidence interval is given to each numerical result. All the simulations have been run by MATLAB 2015a on a computer with 3.2 GHz Intel(R) Core(TM) i5-4690S CPU and 8 GBytes RAM.

\subsection{Uniform Traffic Distribution}
\label{subsec:unid}

We first verify the three scenarios under the uniform traffic distribution: R-U, NSF-U and NJ-U. The corresponding results are demonstrated in Figs. \ref{Fig: nrru}, \ref{Fig: nrnsfu} and \ref{Fig: nrnju} respectively.

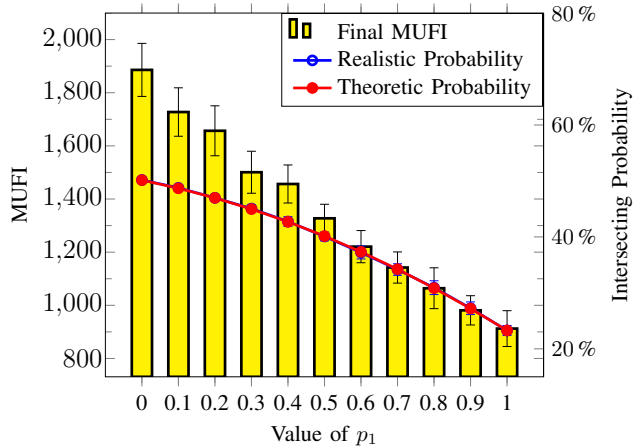
\begin{figure}[!htbp]
\centering
\begin{tikzpicture}[scale=0.85]
\begin{axis}[
        minor tick num =1,
        xlabel=Value of $p_1$,
       minor xtick={1,2,3,4,5,6,7,8,9,10,11},
       xticklabels={0,0.1,0.2,0.3,0.4,0.5,0.6,0.7,0.8,0.9,1},
         ylabel=MUFI,
         scaled y ticks = false,
          ybar,
         xtick=data,
%      nodes near coords,
%      nodes near coords align={vertical},
 yticklabel style = {font=\large,xshift=0.5ex},
]

\addplot[fill=yellow,very thick,ybar=1*\pgflinewidth, bar width=9pt,ybar legend,error bars/.cd, y dir=both, y explicit] coordinates {
        (1,1885.92)  += (0,99.8) -= (0,99.8)
        (2,1727.36)  += (0,91.14) -= (0,91.14)
        (3,1656.78)  += (0,94.00) -= (0,94.00)
        (4,1500.78)  += (0,79.02) -= (0,79.02)
        (5,1456.54)  += (0,71.60) -= (0,71.60)
        (6,1327.36)  += (0,52.93) -= (0,52.93)
        (7,1220.78)  += (0,60.55) -= (0,60.55)
        (8,1142.24)  += (0,58.71) -= (0,58.71)
        (9,1064.28)  += (0,76.88) -= (0,76.88)
        (10,981.10)  += (0,55.07) -= (0,55.07)
        (11,912.42)  += (0,67.38) -= (0,67.38)
        
};\label{plot_one1}
\end{axis}
\begin{axis}[
  axis y line*=right,
  ylabel=Intersecting Probability,
  ymin=0.15,ymax=0.8,
  axis x line=none,
  scaled ticks=false,
  legend style={at={(0.7,1)},anchor= north,legend cell align=left},
        yticklabel=\pgfmathparse{100*\tick}\pgfmathprintnumber{\pgfmathresult}\,\%,
        yticklabel style={/pgf/number format/.cd,fixed,precision=2}  
]
\addlegendimage{/pgfplots/refstyle=plot_one1}\addlegendentry{Final MUFI}
 \addplot[color=blue,mark=o,very thick,error bars/.cd, y dir=both, y explicit] coordinates {
        (1,0.5017)  += (0,0.0043) -= (0,0.0043)
        (2,0.4877)  += (0,0.0049) -= (0,0.0049)
        (3,0.4698)  += (0,0.0061) -= (0,0.0061)
        (4,0.4506)  += (0,0.0069) -= (0,0.0069)
        (5,0.4275)  += (0,0.0086) -= (0,0.0086)
        (6,0.4008)  += (0,0.0078) -= (0,0.0078)
        (7,0.3721)  += (0,0.0115) -= (0,0.0115)
        (8,0.3417)  += (0,0.0103) -= (0,0.0103)
        (9,0.3093)  += (0,0.0123) -= (0,0.0123)
        (10,0.2725)  += (0,0.0113) -= (0,0.0113)
        (11,0.2328)  += (0,0.0089) -= (0,0.0089)
};
\addlegendentry{Realistic Probability};
\addplot[color=red,mark=*,very thick,error bars/.cd, y dir=both, y explicit] coordinates {
        (1,0.5014) 
        (2,0.4869)  
        (3,0.4697)  
        (4,0.4497)  
        (5,0.4270)  
        (6,0.4015)  
        (7,0.3733)  
        (8,0.3423)  
        (9,0.3085)  
        (10,0.2720)  
        (11,0.2328)  
};
\addlegendentry{Theoretic Probability};
\end{axis}
\end{tikzpicture}
\caption{Numerical results for R-U Scenario.}
\label{Fig: nrru}
\end{figure}

\begin{figure}[!htbp]
\centering
\begin{tikzpicture}[scale=0.85]
\begin{axis}[
        minor tick num =1,
        xlabel=Value of $p_1$,
       minor xtick={1,2,3,4,5,6,7,8,9,10,11},
       xticklabels={0,0.1,0.2,0.3,0.4,0.5,0.6,0.7,0.8,0.9,1},
         ylabel=MUFI,
         scaled y ticks = false,
          ybar,
         xtick=data,
%      nodes near coords,
%      nodes near coords align={vertical},
 yticklabel style = {font=\large,xshift=0.5ex},
]

\addplot[fill=yellow,very thick,ybar=1*\pgflinewidth, bar width=9pt,ybar legend,error bars/.cd, y dir=both, y explicit] coordinates {
        (1,508.56)  += (0,65.70) -= (0,65.70)
        (2,479.18)  += (0,56.84) -= (0,56.84)
        (3,450.54)  += (0,49.81) -= (0,49.81)
        (4,423.22)  += (0,54.64) -= (0,54.64)
        (5,400.22)  += (0,47.43) -= (0,47.43)
        (6,370.28)  += (0,40.41) -= (0,40.41)
        (7,352.32)  += (0,42.15) -= (0,42.15)
        (8,342.10)  += (0,42.41) -= (0,42.41)
        (9,332.34)  += (0,45.90) -= (0,45.90)
        (10,332.10)  += (0,47.93) -= (0,47.93)
        (11,330)  += (0,45.01) -= (0,45.01)
        
};\label{plot_one2}
\end{axis}
\begin{axis}[
  axis y line*=right,
  ylabel=Intersecting Probability,
  ymin=0,ymax=0.5,
  axis x line=none,
  scaled ticks=false,
  legend style={at={(0.7,1)},anchor= north,legend cell align=left},
        yticklabel=\pgfmathparse{100*\tick}\pgfmathprintnumber{\pgfmathresult}\,\%,
        yticklabel style={/pgf/number format/.cd,fixed,precision=2}  
]
\addlegendimage{/pgfplots/refstyle=plot_one2}\addlegendentry{Final MUFI}
 \addplot[color=blue,mark=o,very thick,error bars/.cd, y dir=both, y explicit] coordinates {
        (1,0.2035)  += (0,0.0075) -= (0,0.0075)
        (2,0.1905)  += (0,0.0070) -= (0,0.0070)
        (3,0.1779)  += (0,0.0074) -= (0,0.0074)
        (4,0.1658)  += (0,0.0073) -= (0,0.0073)
        (5,0.1551)  += (0,0.0058) -= (0,0.0058)
        (6,0.1436)  += (0,0.0055) -= (0,0.0055)
        (7,0.1333)  += (0,0.0065) -= (0,0.0065)
        (8,0.1242)  += (0,0.0060) -= (0,0.0060)
        (9,0.1150)  += (0,0.0050) -= (0,0.0050)
        (10,0.1062)  += (0,0.0048) -= (0,0.0048)
        (11,0.0977)  += (0,0.0046) -= (0,0.0046)
};
\addlegendentry{Realistic Probability};
\addplot[color=red,mark=*,very thick,error bars/.cd, y dir=both, y explicit] coordinates {
        (1,0.2042) 
        (2,0.1912)  
        (3,0.1787)  
        (4,0.1667)  
        (5,0.1553)  
        (6,0.1444)  
        (7,0.1340)  
        (8,0.1242)  
        (9,0.1149)  
        (10,0.1062)  
        (11,0.0979)  
};
\addlegendentry{Theoretic Probability};
\end{axis}
\end{tikzpicture}
\caption{Numerical results for NSF-U Scenario.}
\label{Fig: nrnsfu}
\end{figure}
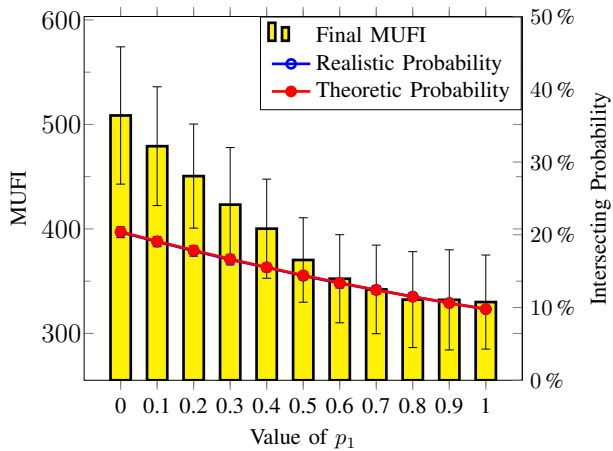

\begin{figure}[!htbp]
\centering
\begin{tikzpicture}[scale=0.85]
\begin{axis}[
        minor tick num =1,
        xlabel=Value of $p_1$,
       minor xtick={1,2,3,4,5,6,7,8,9,10,11},
       ymin=0,ymax=600,
       xticklabels={0,0.1,0.2,0.3,0.4,0.5,0.6,0.7,0.8,0.9,1},
         ylabel=MUFI,
         scaled y ticks = false,
          ybar,
         xtick=data,
%      nodes near coords,
%      nodes near coords align={vertical},
 yticklabel style = {font=\large,xshift=0.5ex},
]

\addplot[fill=yellow,very thick,ybar=1*\pgflinewidth, bar width=9pt,ybar legend,error bars/.cd, y dir=both, y explicit] coordinates {
        (1,362.92)  += (0,54.10) -= (0,54.10)
        (2,359.26)  += (0,58.44) -= (0,58.44)
        (3,357.38)  += (0,63.17) -= (0,63.17)
        (4,360.20)  += (0,58.83) -= (0,58.83)
        (5,358.74)  += (0,51.86) -= (0,51.86)
        (6,363.74)  += (0,51.31) -= (0,51.31)
        (7,361.34)  += (0,58.10) -= (0,58.10)
        (8,374.58)  += (0,50.54) -= (0,50.54)
        (9,365.56)  += (0,51.57) -= (0,51.57)
        (10,372.44)  += (0,52.71) -= (0,52.71)
        (11,379.06)  += (0,52.84) -= (0,52.84)
        
};\label{plot_one3}
\end{axis}
\begin{axis}[
  axis y line*=right,
  ylabel=Intersecting Probability,
  ymin=0,ymax=0.5,
  axis x line=none,
  scaled ticks=false,
  legend style={at={(0.7,1)},anchor= north,legend cell align=left},
        yticklabel=\pgfmathparse{100*\tick}\pgfmathprintnumber{\pgfmathresult}\,\%,
        yticklabel style={/pgf/number format/.cd,fixed,precision=2}  
]
\addlegendimage{/pgfplots/refstyle=plot_one3}\addlegendentry{Final MUFI}
 \addplot[color=blue,mark=o,very thick,error bars/.cd, y dir=both, y explicit] coordinates {
        (1,0.1155)  += (0,0.0052) -= (0,0.0052)
        (2,0.1101)  += (0,0.0050) -= (0,0.0050)
        (3,0.1052)  += (0,0.0052) -= (0,0.0052)
        (4,0.1006)  += (0,0.0052) -= (0,0.0052)
        (5,0.0971)  += (0,0.0052) -= (0,0.0052)
        (6,0.0944)  += (0,0.0061) -= (0,0.0061)
        (7,0.0920)  += (0,0.0055) -= (0,0.0055)
        (8,0.0909)  += (0,0.0064) -= (0,0.0064)
        (9,0.0905)  += (0,0.0072) -= (0,0.0072)
        (10,0.0904)  += (0,0.0066) -= (0,0.0066)
        (11,0.0909)  += (0,0.0070) -= (0,0.0070)
};
\addlegendentry{Realistic Probability};
\addplot[color=red,mark=*,very thick,error bars/.cd, y dir=both, y explicit] coordinates {
        (1,0.1157) 
        (2,0.1100)  
        (3,0.1049)  
        (4,0.1006)  
        (5,0.0970)  
        (6,0.0940)  
        (7,0.0918)  
        (8,0.0903)  
        (9,0.0895)  
        (10,0.0895)  
        (11,0.0901)  
};
\addlegendentry{Theoretic Probability};
\end{axis}
\end{tikzpicture}
\caption{Numerical results for NJ-U Scenario.}
\label{Fig: nrnju}
\end{figure}
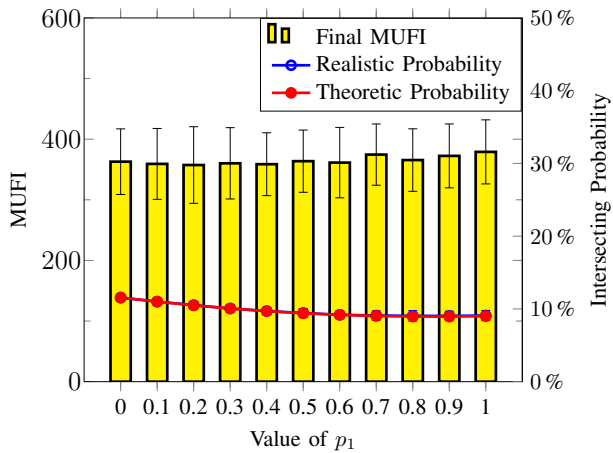

From these results, we can observe that in all the three scenarios, the realistic intersecting probabilities (marked by blue lines) perfectly fit in with the theoretic ones (marked by red lines) that we can barely see the blue lines in the three figures. 

The results of intersecting probabilities prove the effectiveness of  our Theoretic  Chain 2 in Fig. \ref{fig:chian the2}, \textit{i.e.}, the conflict coefficients and the computing method of GOF.

For R-U and NSF-U, with the intersecting probabilities gradually decreasing, the corresponding MUFIs reduce and reach the minimum when $p_1=1$ in Figs. \ref{Fig: nrru} and \ref{Fig: nrnsfu}, $\textit{i.e.}$, their own optimal routing schemes. For NJ-U in Fig \ref{Fig: nrnju}, as we can observe that the intersecting probability $p$ is within a narrow range of $[8.94\%, 11.57\%]$  which means the difference of intersecting probabilities is less than $3\%$. Thus, the difference of the corresponding conflict graphs is so subtle that only a small volatility of MUFIs ([357.38, 379.06]) is observed.

This results of MUFIs prove the effectiveness of our Theoretic Chain 1 in Fig. \ref{fig:chian the}, \textit{i.e.}, the intersecting probability is positively  correlated to the final MUFI. 

%Further, the results in NJ-U also manifest that which needs further inves small difference among intersecting probabilities.  

\subsection{Weighted Traffic Distribution}
\label{subsec:concD}

Here, we evaluate the other three scenarios with the weighted traffic distribution: R-W, NSF-W and NJ-W. The corresponding results are demonstrated in Figs. \ref{Fig: nrrc}, \ref{Fig: nrnsfc} and \ref{Fig: nrnjc} respectively.
  
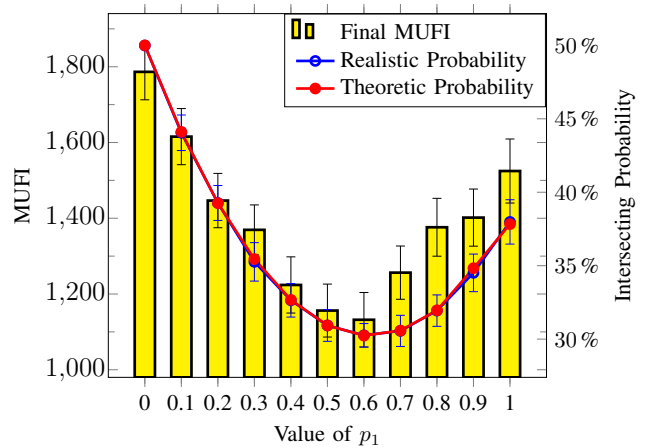
\begin{figure}[!htbp]
\centering
\begin{tikzpicture}[scale=0.85]
\begin{axis}[
        minor tick num =1,
        xlabel=Value of $p_1$,
       minor xtick={1,2,3,4,5,6,7,8,9,10,11},
       xticklabels={0,0.1,0.2,0.3,0.4,0.5,0.6,0.7,0.8,0.9,1},
         ylabel=MUFI,
         scaled y ticks = false,
          ybar,
         xtick=data,
%      nodes near coords,
%      nodes near coords align={vertical},
 yticklabel style = {font=\large,xshift=0.5ex},
]

\addplot[fill=yellow,very thick,ybar=1*\pgflinewidth, bar width=9pt,ybar legend,error bars/.cd, y dir=both, y explicit] coordinates {
        (1,1786.5)  += (0,73.78) -= (0,73.78)
        (2,1615.5)  += (0,74.10) -= (0,74.10)
        (3,1446.7)  += (0,71.72) -= (0,71.72)
        (4,1369.4)  += (0,65.80) -= (0,65.80)
        (5,1223.9)  += (0,74.19) -= (0,74.19)
        (6,1156.28)  += (0,69.90) -= (0,69.90)
        (7,1132)  += (0,72.03) -= (0,72.03)
        (8,1256.4)  += (0,70.30) -= (0,70.30)
        (9,1376.2)  += (0,76.38) -= (0,76.38)
        (10,1401.7)  += (0,75.42) -= (0,75.42)
        (11,1524.7)  += (0,84.53) -= (0,84.53)
        
};\label{plot_one}
\end{axis}
\begin{axis}[
  axis y line*=right,
  ylabel=Intersecting Probability,
  axis x line=none,
  scaled ticks=false,
  legend style={at={(0.7,1)},anchor= north,legend cell align=left},
        yticklabel=\pgfmathparse{100*\tick}\pgfmathprintnumber{\pgfmathresult}\,\%,
        yticklabel style={/pgf/number format/.cd,fixed,precision=2}  
]
\addlegendimage{/pgfplots/refstyle=plot_one}\addlegendentry{Final MUFI}
 \addplot[color=blue,mark=o,very thick,error bars/.cd, y dir=both, y explicit] coordinates {
        (1,0.4999)  += (0,0.0011) -= (0,0.0011)
        (2,0.4405)  += (0,0.0121) -= (0,0.0121)
        (3,0.3927)  += (0,0.0119) -= (0,0.0119)
        (4,0.3526)  += (0,0.0131) -= (0,0.0131)
        (5,0.3264)  += (0,0.0114) -= (0,0.0114)
        (6,0.3093)  += (0,0.0108) -= (0,0.0108)
        (7,0.3026)  += (0,0.0080) -= (0,0.0080)
        (8,0.3056)  += (0,0.0106) -= (0,0.0106)
        (9,0.3194)  += (0,0.0107) -= (0,0.0107)
        (10,0.3451)  += (0,0.0128) -= (0,0.0128)
        (11,0.3798)  += (0,0.0151) -= (0,0.0151)
};
\addlegendentry{Realistic Probability};
\addplot[color=red,mark=*,very thick,error bars/.cd, y dir=both, y explicit] coordinates {
        (1,0.5) 
        (2,0.4411)  
        (3,0.3926)  
        (4,0.3545)  
        (5,0.3267)  
        (6,0.3093)  
        (7,0.3024)  
        (8,0.3058)  
        (9,0.3196)  
        (10,0.3483)  
        (11,0.3783)  
};
\addlegendentry{Theoretic Probability};
\end{axis}
\end{tikzpicture}
\caption{Numerical results for R-W Scenario.}
\label{Fig: nrrc}
\end{figure}

\begin{figure}[!htbp]
\centering
\begin{tikzpicture}[scale=0.85]
\begin{axis}[
        minor tick num =1,
        xlabel=Value of $p_1$,
       minor xtick={1,2,3,4,5,6,7,8,9,10,11},
       xticklabels={0,0.1,0.2,0.3,0.4,0.5,0.6,0.7,0.8,0.9,1},
         ylabel=MUFI,
         scaled y ticks = false,
          ybar,
         xtick=data,
%      nodes near coords,
%      nodes near coords align={vertical},
 yticklabel style = {font=\large,xshift=0.5ex},
]

\addplot[fill=yellow,very thick,ybar=1*\pgflinewidth, bar width=9pt,ybar legend,error bars/.cd, y dir=both, y explicit] coordinates {
        (1,1500.5)  += (0,85.25) -= (0,85.25)
        (2,1420.5)  += (0,91.10) -= (0,91.10)
        (3,1380)  += (0,89.22) -= (0,89.22)
        (4,1350)  += (0,89.12) -= (0,89.12)
        (5,1320)  += (0,91.23) -= (0,91.23)
        (6,1319)  += (0,89.32) -= (0,89.32)
        (7,1320)  += (0,91.8) -= (0,91.8)
        (8,1353)  += (0,89.9) -= (0,89.9)
        (9,1374)  += (0,91.3) -= (0,91.3)
        (10,1401.7)  += (0,84.42) -= (0,84.42)
        (11,1452.6)  += (0,73.16) -= (0,73.16)
        
};\label{plot_one}
\end{axis}
\begin{axis}[
  axis y line*=right,
  ylabel=Intersecting Probability,
  axis x line=none,
  scaled ticks=false,
  legend style={at={(0.7,1)},anchor= north,legend cell align=left},
        yticklabel=\pgfmathparse{100*\tick}\pgfmathprintnumber{\pgfmathresult}\,\%,
        yticklabel style={/pgf/number format/.cd,fixed,precision=2}  
]
\addlegendimage{/pgfplots/refstyle=plot_one}\addlegendentry{Final MUFI}
 \addplot[color=blue,mark=o,very thick,error bars/.cd, y dir=both, y explicit] coordinates {
        (1,0.3976)  += (0,0.0189) -= (0,0.00189)
        (2,0.3631)  += (0,0.0188) -= (0,0.0188)
        (3,0.3338)  += (0,0.0158) -= (0,0.0158)
        (4,0.3141)  += (0,0.0161) -= (0,0.0161)
        (5,0.2995)  += (0,0.0147) -= (0,0.0147)
        (6,0.2923)  += (0,0.0141) -= (0,0.0141)
        (7,0.2912)  += (0,0.0140) -= (0,0.0140)
        (8,0.2971)  += (0,0.0176) -= (0,0.0176)
        (9,0.3108)  += (0,0.0165) -= (0,0.0165)
        (10,0.3301)  += (0,0.0149) -= (0,0.0149)
        (11,0.3548)  += (0,0.0166) -= (0,0.0166)
};
\addlegendentry{Realistic Probability};
\addplot[color=red,mark=*,very thick,error bars/.cd, y dir=both, y explicit] coordinates {
        (1,0.3982) 
        (2,0.3642)  
        (3,0.3368)  
        (4,0.3161)  
        (5,0.3019)  
        (6,0.2943)  
        (7,0.2934)  
        (8,0.2990)  
        (9,0.3112)  
        (10,0.3300)  
        (11,0.3554)  
};
\addlegendentry{Theoretic Probability};
\end{axis}
\end{tikzpicture}
\caption{Numerical results for NSF-W Scenario.}
\label{Fig: nrnsfc}
\end{figure}
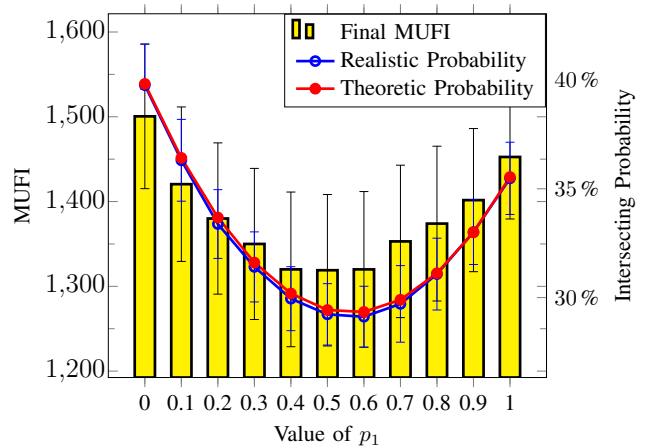

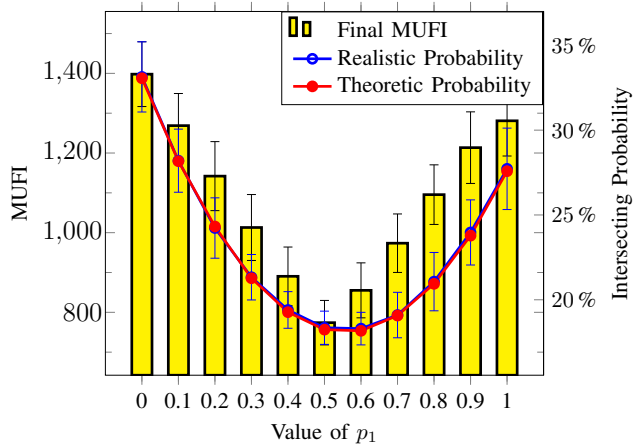
\begin{figure}[!htbp]
\centering
\begin{tikzpicture}[scale=0.85]
\begin{axis}[
        minor tick num =1,
        xlabel=Value of $p_1$,
       minor xtick={1,2,3,4,5,6,7,8,9,10,11},
       xticklabels={0,0.1,0.2,0.3,0.4,0.5,0.6,0.7,0.8,0.9,1},
         ylabel=MUFI,
         scaled y ticks = false,
          ybar,
         xtick=data,
%      nodes near coords,
%      nodes near coords align={vertical},
 yticklabel style = {font=\large,xshift=0.5ex},
]

\addplot[fill=yellow,very thick,ybar=1*\pgflinewidth, bar width=9pt,ybar legend,error bars/.cd, y dir=both, y explicit] coordinates {
        (1,1398)  += (0,81.3) -= (0,81.3)
        (2,1268.8)  += (0,80.65) -= (0,80.65)
        (3,1142)  += (0,86.65) -= (0,86.65)
        (4,1013)  += (0,82.77) -= (0,82.77)
        (5,890.44)  += (0,73.50) -= (0,73.50)
        (6,774.04)  += (0,55.83) -= (0,55.83)
        (7,855.2)  += (0,69.01) -= (0,69.01)
        (8,973.58)  += (0,73.44) -= (0,73.44)
        (9,1095.5)  += (0,74.83) -= (0,74.83)
        (10,1213.5)  += (0,89.89) -= (0,89.89)
        (11,1280.9)  += (0,88.39) -= (0,88.39)
        
};\label{plot_one}
\end{axis}
\begin{axis}[
  axis y line*=right,
  ylabel=Intersecting Probability,
  axis x line=none,
  scaled ticks=false,
  legend style={at={(0.7,1)},anchor= north,legend cell align=left},
        yticklabel=\pgfmathparse{100*\tick}\pgfmathprintnumber{\pgfmathresult}\,\%,
        yticklabel style={/pgf/number format/.cd,fixed,precision=2}  
]
\addlegendimage{/pgfplots/refstyle=plot_one}\addlegendentry{Final MUFI}
 \addplot[color=blue,mark=o,very thick,error bars/.cd, y dir=both, y explicit] coordinates {
        (1,0.3314)  += (0,0.0207) -= (0,0.0207)
        (2,0.2820)  += (0,0.0186) -= (0,0.0186)
        (3,0.2423)  += (0,0.0178) -= (0,0.0178)
        (4,0.2133)  += (0,0.0134) -= (0,0.0134)
        (5,0.1940)  += (0,0.0108) -= (0,0.0108)
        (6,0.1835)  += (0,0.0098) -= (0,0.0098)
        (7,0.1830)  += (0,0.0096) -= (0,0.0096)
        (8,0.1910)  += (0,0.0134) -= (0,0.0134)
        (9,0.2107)  += (0,0.0172) -= (0,0.0172)
        (10,0.2397)  += (0,0.0192) -= (0,0.0192)
        (11,0.2772)  += (0,0.0240) -= (0,0.0240)
};
\addlegendentry{Realistic Probability};
\addplot[color=red,mark=*,very thick,error bars/.cd, y dir=both, y explicit] coordinates {
        (1,0.3306) 
        (2,0.2817)  
        (3,0.2432)  
        (4,0.2127)  
        (5,0.1927)  
        (6,0.1824)  
        (7,0.1817)  
        (8,0.1907)  
        (9,0.2094)  
        (10,0.2378)  
        (11,0.2758)  
};
\addlegendentry{Theoretic Probability};
\end{axis}
\end{tikzpicture}
\caption{Numerical results for NJ-W Scenario.}
\label{Fig: nrnjc}
\end{figure}

Similarly, the realistic intersecting probabilities match very well with the theoretical ones in the three scenarios, which again prove the effectiveness of Theoretic Chain 2 of the conflict coefficients and GOF.   

From the aspect of final MUFIs, the three scenarios of R-W, NSF-W and  NJ-W represent a common characteristic: First, with the decreasing of the intersecting probabilities, the corresponding MUFIs reduce. After passing their optimal routing schemes, the corresponding MUFIs keep increasing as the intersecting probabilities grow. These results further verify the effectiveness of Theoretic Chain 1 of the intersecting probability. 

Meanwhile, the results also exhibited the importance of designing a better routing scheme to decrease the intersecting probability. When the information of the network topology and traffic distribution are obtained, how to optimally assign the requests on the $K$ candidate paths to decrease the intersecting probability is crucial to the final performance of the RSA. Taking the three scenarios of R-W, NSF-W and  NJ-W with $K=2$ for examples, the (worst, best) MUFI pairs are (1786.51, 1132.32), (1500.23, 1319.02) and (1398.41, 774.08) respectively. Thus, decreasing the intersecting probability can obtain a huge gain in the final spectrum usage.

\subsection{Comparisons in the Frame of Intersecting Probability}

In this subsection, we compare the minimum intersecting probability of the six scenarios using their own optimal routing schemes in Table \ref{tab: tmipiss}. Besides, we present the numerical results of all the 66 cases in the panoramic Fig. \ref{Fig: tmufi66} with intersecting probabilities as the X-axis and MUFIs as the Y-axis to show clearly their correlation.

\begin{figure}[!htbp]
\centering
\begin{tikzpicture}[scale=0.85]
\begin{axis}[
        minor tick num =1,
        %xlabel=Six,
       minor xtick={1,2,3,4,5,6},
       xticklabels={NSF-U,NJ-U, NJ-W,R-U,R-W,NSF-W},
         ylabel=MUFI,
         scaled y ticks = false,
          ybar,
         xtick=data,
%      nodes near coords,
%      nodes near coords align={vertical},
 yticklabel style = {font=\large,xshift=0.5ex},
]

\addplot[fill=yellow,very thick,ybar=1*\pgflinewidth, bar width=9pt,ybar legend,error bars/.cd, y dir=both, y explicit] coordinates {
        (1,330)  += (0,45.01) -= (0,45.01)
        (2,365.56)  += (0,51.57) -= (0,51.57)
        (3,774.04)  += (0,55.83) -= (0,55.83)
        (4,912.42)  += (0,67.38) -= (0,67.38)
        (5,1157)  += (0,72.03) -= (0,72.03)
        (6,1319)  += (0,112.32) -= (0,112.32)

};\label{plot_one}
\end{axis}
\begin{axis}[
  axis y line*=right,
  ylabel=Intersecting Probability,
  axis x line=none,
  scaled ticks=false,
  legend style={at={(0.3,1)},anchor= north,legend cell align=left},
        yticklabel=\pgfmathparse{100*\tick}\pgfmathprintnumber{\pgfmathresult}\,\%,
        yticklabel style={/pgf/number format/.cd,fixed,precision=2}  
]
\addlegendimage{/pgfplots/refstyle=plot_one}\addlegendentry{Final MUFI}
 \addplot[color=blue,mark=o,very thick,error bars/.cd, y dir=both, y explicit] coordinates {
        (1,0.0977)  += (0,0.0046) -= (0,0.0046)
        (2,0.0909)  += (0,0.0064) -= (0,0.0064)
        (3,0.1830)  += (0,0.0096) -= (0,0.0096)
        (4,0.2328)  += (0,0.0089) -= (0,0.0089)
        (5,0.3025)  += (0,0.0147) -= (0,0.0147)
        (6,0.2923)  += (0,0.0141) -= (0,0.0141)

};
\addlegendentry{Realistic Probability};
\addplot[color=red,mark=*,very thick,error bars/.cd, y dir=both, y explicit] coordinates {
        (1,0.0979) 
        (2,0.0894)  
        (3,0.1808)  
        (4,0.2328)  
        (5,0.3026)  
        (6,0.2930)  

};
\addlegendentry{Theoretic Probability};
\end{axis}
\end{tikzpicture}
\caption{Numerical results for the six scenarios using their optimal routing schemes.}
\label{Fig: nrnssopr}
\end{figure}
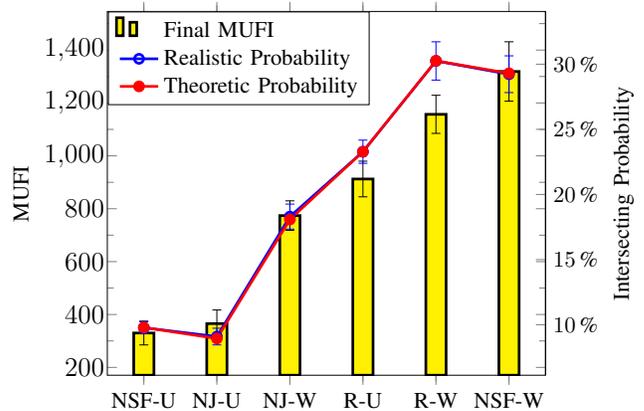

In Fig. \ref{Fig: nrnssopr}, we present the final MUFIs of the six scenarios using their own optimal routing schemes. Compared to the results in Figs. \ref{Fig: nrru}-\ref{Fig: nrnjc}, it can be observed that the MUFI of the optimal routing scheme of each scenario is indeed the minimum. The realistic intersecting probability, as always, matches very well the theoretic one, which once again proves the effectiveness of the Theoretic Chain 2 of the conflict coefficients and GOF. 

From Table \ref{tab: tmipiss}, the increasing order of minimum intersecting probabilities of the six scenarios is as follows.

\begin{table}[!htb]
\centering
  \caption{The Increasing Order of Minimum Intersecting Probabilities}
\label{tab: tiosminip}
\resizebox{0.47\textwidth}{!}{\begin{tabular}{|c|c|c|c|c|c|}
\hline
 NJ-U & NSF-U &  NJ-W & R-U & NSF-W & R-W \\
\hline
  8.94\% &  9.79\%  & 18.08\% & 23.28\% & 29.30\% & 30.26\%\\
\hline
\end{tabular}}
\end{table}

From the final MUFIs of Fig. \ref{Fig: nrnssopr}, we can see that the intersecting probabilities are in general positively correlated to the final MUFIs except two cases: NSF-U vs. NJ-U and NSF-W vs. R-W. The two exceptions can be interpreted as follows. We can see that the intersecting probability difference between NSF-U and NJ-U is less than 1\% (similar for R-W and NSF-W). Thus, the final number of edges of their conflict graphs are so close that the spectrum assignment plays a significant role in the final MUFIs, while the performances of the spectrum assignment in different EONs are varying.

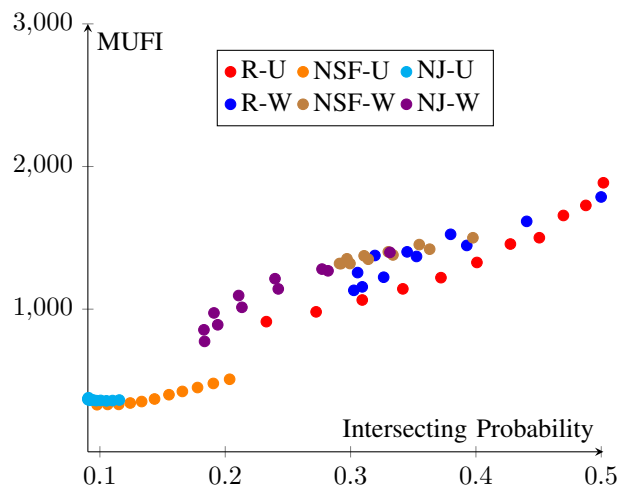
\begin{figure}[!htbp]
\begin{tikzpicture}
\begin{axis}[
    legend style={at={(0.25,0.85)},anchor= west,legend cell align=left,legend columns=3},
    axis lines=middle,
    ymin=0, ymax=3000,
    ylabel=MUFI, 
    xlabel=Intersecting Probability,
%     width=0.5\textwidth,
%     height=0.3\textwidth,
]

\addplot [only marks,red] table {
  0.5017 1885.92
  0.4877 1727.36  
  0.4698 1656.78  
  0.4506 1500.78 
  0.4275 1456.54
  0.4008 1327.36  
  0.3721 1220.78  
  0.3417 1142.24  
  0.3093 1064.28  
  0.2725 981.10  
  0.2328 912.42                          
};\addlegendentry{R-U};
\addplot [only marks, color=orange] table {
 0.2035 508.56
 0.1905 479.18 
 0.1779 450.54
 0.1658 423.22  
 0.1551 400.22 
 0.1436 370.28  
 0.1333 352.32 
 0.1242 342.10  
 0.1150 332.34
 0.1062 332.10
 0.0977 330                    
};\addlegendentry{NSF-U};
\addplot [only marks, color=cyan] table {
    0.1155 362.92
    0.1101 359.26
    0.1052 357.38
    0.1006 360.20
    0.0971 358.74
    0.0944 363.74
    0.0920 361.34
    0.0909 374.58
    0.0905 365.56
    0.0904 372.44
    0.0909 379.06   
};\addlegendentry{NJ-U};
\addplot [only marks, color=blue] table {
     0.4999 1786.5
     0.4405 1615.5
     0.3927 1446.7
     0.3526 1369.4
     0.3264 1223.9
     0.3093 1156.28
     0.3026 1132
     0.3056 1256.4
     0.3194 1376.2
     0.3451 1401.7
     0.3798 1524.7       
};\addlegendentry{R-W};
\addplot [only marks, color=brown] table {
   0.3976 1500.5
   0.3631 1420.5
   0.3338 1380
   0.3141 1350
   0.2995 1320
   0.2923 1319
   0.2912 1320
   0.2971 1353
   0.3108 1374
   0.3301 1401.7
   0.3548 1452.6           
};\addlegendentry{NSF-W};

\addplot [only marks, color=violet] table {
   0.3314 1398
   0.2820 1268.8
   0.2423 1142
   0.2133 1013
   0.1940 890.44
   0.1835 774.04
   0.1830 855.2
   0.1910 973.58
   0.2107 1095.5
   0.2397 1213.5
   0.2772 1280.9           
};\addlegendentry{ NJ-W};
\end{axis}
\end{tikzpicture}
\caption{The MUFIs and intersecting probabilities in the 66 cases.}
\label{Fig: tmufi66}
\end{figure}

To give a panoramic picture of the relation between intersecting probabilities and final MUFIs, we collect together in Fig. \ref{Fig: tmufi66} the numerical results of the 66 cases. Although the 66 cases are conducted in different EONs, traffic distributions as well as routing schemes, we can uniformly analyze their differences in terms of the intersecting probability. Figure \ref{Fig: tmufi66} further confirms the Theoretical Chain 1 in Fig. \ref{fig:chian the}.

In summary, all the numerical results validated the proposed Theoretical Chain 1 in Fig. \ref{fig:chian the} and Theoretical Chain 2 in Fig. \ref{fig:chian the2}. With the help of the key role, intersecting probability, the two theoretical chains exactly figure out how the network topology, traffic distribution and routing scheme impact on the spectrum usage. These numerical results also demonstrate the importance of decreasing the intersecting probability. 

\section{Conclusions}
\label{sec:conclusion}
In this work, we provided a theoretical analysis to reveal how network topology, traffic distribution and routing scheme impact on the spectrum usage. In this theoretical analysis, we developed two theoretical chains to figure out the implicit impact mechanism. 

We first investigated the property of the conflict graph built upon the computed lightpaths. We proved that the optimal MUFI of the conflict graph is directly determined by its chromatic number, and the chromatic number has a strongly positive correlation to the edge existence probability in the conflict graph, \textit{i.e.}, the intersecting probability. In other words, the smaller the intersection probability, the smaller the optimal MUFI, which constitutes our first theoretical chain. 

We then proposed the concept of conflict coefficients, which are important parameters decided by the network topology and traffic distribution. We further developed the quadratic programming GOF with the routing scheme to determine the intersecting probability. This constitutes our second theoretical chain.  

Finally, the proposed theoretical chains have been validated by extensive simulations in several well-known EONs. The future work about the application of this work can go in-depth in the two directions: 1. To adjust the network topology and traffic distribution to optimize the conflict coefficients so as to decrease the minimum intersecting probability; 2. to take into account the network topology and traffic distribution to plan a good routing scheme to approach the minimum intersecting probability.  

%Next, by leveraging random graph theory, we introduced the intersecting probability. The Theoretical Chain 1 in Fig. \ref{fig:chian the} first theoretically unveils the connection between the intersecting probability and the optimal MUFI. The smaller the intersecting probability, the smaller optimal MUFI. Next, we then proposed the concept of conflict coefficient and a quadratic programming GOF based on the traffic distribution and network topology  and come up with the Theoretical Chain 2 in Fig. \ref{fig:chian the2}. Combing the two theoretical chains, we figured out the impact of the traffic distribution and network topology. We computed  the corresponding conflict coefficients and GOFs of three EONs under two different traffic distributions. Finally, we conducted six scenarios with 66 cases of simulations to verify our derivation. Furthermore, according to the analysis in this paper, the information of traffic distribution and network topology is crucial to the final performances of EONs. Therefore, to improve the final performances, the routing phase of EONs in the future should be more intelligent based on the prediction of traffic distribution and the underlying topology.
\bibliographystyle{IEEEtran}
\bibliography{reference}

\begin{thebibliography}{10}
\providecommand{\url}[1]{#1}
\csname url@samestyle\endcsname
\providecommand{\newblock}{\relax}
\providecommand{\bibinfo}[2]{#2}
\providecommand{\BIBentrySTDinterwordspacing}{\spaceskip=0pt\relax}
\providecommand{\BIBentryALTinterwordstretchfactor}{4}
\providecommand{\BIBentryALTinterwordspacing}{\spaceskip=\fontdimen2\font plus
\BIBentryALTinterwordstretchfactor\fontdimen3\font minus
  \fontdimen4\font\relax}
\providecommand{\BIBforeignlanguage}[2]{{%
\expandafter\ifx\csname l@#1\endcsname\relax
\typeout{** WARNING: IEEEtran.bst: No hyphenation pattern has been}%
\typeout{** loaded for the language `#1'. Using the pattern for}%
\typeout{** the default language instead.}%
\else
\language=\csname l@#1\endcsname
\fi
#2}}
\providecommand{\BIBdecl}{\relax}
\BIBdecl

\bibitem{Gerstel2012}
O.~Gerstel, M.~Jinno, A.~Lord, and B.~Yoo, ``Elastic optical networking: a new
  dawn for the optical layer?'' \emph{{IEEE} Commun. Mag.}, vol.~50, pp.
  s12--s20, Feb. 2012.

\bibitem{s7:b1}
P.~Lu, L.~Zhang, X.~Liu, J.~Yao, and Z.~Zhu, ``Highly efficient data migration
  and backup for big data applications in elastic optical inter-data-center
  networks,'' \emph{IEEE Netw.}, vol.~29, pp. 36--42, Sept./Oct. 2015.

\bibitem{s6:b4}
M.~Jinno, H.~Takara, B.~Kozicki, Y.~Tsukishima, Y.~Sone, and S.~Matsuoka,
  ``Spectrum-efficient and scalable elastic optical path network: architecture,
  benefits, and enabling technologies,'' \emph{IEEE Commun. Mag.}, vol.~47, pp.
  66--73, Nov. 2009.

\bibitem{Zhu2013_JLT}
Z.~Zhu, W.~Lu, L.~Zhang, and N.~Ansari, ``Dynamic service provisioning in
  elastic optical networks with hybrid single-/multi-path routing,'' \emph{J.
  Lightw. Technol.}, vol.~31, pp. 15--22, Jan. 2013.

\bibitem{s6:b5}
G.~Zhang, M.~Leenheer, A.~Morea, and B.~Mukherjee, ``A survey on {OFDM}-based
  elastic core optical networking,'' \emph{IEEE Commun. Surveys Tuts.},
  vol.~15, pp. 65--87, First Quarter 2013.

\bibitem{s7:b2}
F.~Ji, X.~Chen, W.~Lu, J.~Rodrigues, and Z.~Zhu, ``Dynamic p-cycle protection
  in spectrum-sliced elastic optical networks,'' \emph{J. Lightw. Technol.},
  vol.~32, pp. 1190--1199, Mar. 2014.

\bibitem{s7:b3}
W.~Fang, M.~Lu, X.~Liu, L.~Gong, and Z.~Zhu, ``Joint defragmentation of optical
  spectrum and {IT} resources in elastic optical datacenter interconnections,''
  \emph{J. Opt. Commun. Netw.}, vol.~7, pp. 314--324, Apr. 2015.

\bibitem{s1:b4}
K.~Christodoulopoulos, I.~Tomkos, and E.~Varvarigos, ``Elastic bandwidth
  allocation in flexible {OFDM}-based optical networks,'' \emph{J. Lightw.
  Technol.}, vol.~29, pp. 1354--1366, May. 2011.

\bibitem{Gong2012_CL}
L.~Gong, X.~Zhou, W.~Lu, and Z.~Zhu, ``A two-population based evolutionary
  approach for optimizing routing, modulation and spectrum assignments {(RMSA)}
  in {O-OFDM} networks,'' \emph{IEEE Commun. Lett.}, vol.~16, pp. 1520--1523,
  Sept. 2012.

\bibitem{ss4:b6}
B.~Jaumard and M.~Daryalal, ``Efficient spectrum utilization in large scale
  {RWA} problems,'' \emph{IEEE/ACM Trans. Netw.}, vol.~25, no.~2, pp.
  1263--1278, Apr. 2017.

\bibitem{ss4:b7}
K.~Walkowiak, R.~Goścień, M.~Klinkowski, and M.~Woźniak, ``Optimization of
  multicast traffic in elastic optical networks with distance-adaptive
  transmission,'' \emph{IEEE Commun. Lett.}, vol.~18, no.~12, pp. 2117--2120,
  Dec. 2014.

\bibitem{ss4:b8}
M.~Klinkowski and K.~Walkowiak, ``Routing and spectrum assignment in spectrum
  sliced elastic optical path network,'' \emph{IEEE Commun. Lett.}, vol.~15,
  no.~8, pp. 884--886, Aug. 2011.

\bibitem{ss4:b1}
R.~Ramamurthy and B.~Mukherjee, ``Fixed-alternate routing and wavelength
  conversion in wavelength-routed optical networks,'' \emph{IEEE/ACM Trans.
  Netw.}, vol.~10, no.~3, pp. 351--367, Jun. 2002.

\bibitem{ss4:b2}
B.~Mukherjee, \emph{Optical WDM Networks}.\hskip 1em plus 0.5em minus
  0.4em\relax Berlin, Germany: SpringerVerlag, 2006.

\bibitem{s1:b2}
Y.~Wang, X.~Cao, Q.~Hu, and Y.~Pan, ``Towards elastic and fine-granular
  bandwidth allocation in spectrum-sliced optical networks,'' \emph{J. Opt.
  Commun. Netw.}, vol.~4, pp. 906--917, Nov. 2012.

\bibitem{s6:b1}
B.~Chatterjee, N.~Sarma, and E.~Oki, ``Routing and spectrum allocation in
  elastic optical networks: A tutorial,'' \emph{IEEE Commun. Surveys Tuts.},
  vol.~17, pp. 1776--1800, Third Quarter 2015.

\bibitem{s7:b4}
M.~Jinno, B.~Kozicki, H.~Takara, A.~Watanabe, Y.~Sone, T.~Tanaka, and
  A.~Hirano, ``Distance-adaptive spectrum resource allocation in
  spectrum-sliced elastic optical path network,'' \emph{IEEE Commun. Mag.},
  vol.~48, pp. 138--145, Aug. 2010.

\bibitem{sss:b1}
H.~Zang, J.~Jue, and B.~Mukherjee, ``A review of routing and wavelength
  assignment approaches for wavelength-routed optical {WDM} networks,''
  \emph{Opt. Netw. Mag.}, pp. 47--60, 2000.

\bibitem{htwuTON2017}
H.~Wu, F.~Zhou, Z.~Zhu, and Y.~Chen, ``On the distance spectrum assignment in
  elastic optical networks,'' \emph{{IEEE/ACM} Trans. Netw.}, vol.~25, no.~4,
  pp. 2391--2404, 2017.

\bibitem{ss4:b9}
L.~Velasco, M.~Klinkowski, M.~Ruiz, and J.~Comellas, ``Modeling the routing and
  spectrum allocation problem for flexgrid optical networks,'' \emph{Photonic
  Netw. Commun.}, vol.~24, pp. 177--186, 2013.

\bibitem{ss4:b10}
X.~Zhou, W.~Lu, L.~Gong, and Z.~Zhu, ``Dynamic rmsa in elastic optical networks
  with an aadaptive genetic algorithm,'' in \emph{Proc. of GLOBECOM 2012}, Dec.
  2012, p. 2912–2917.

\bibitem{s2:b1}
Y.~Wang, X.~Cao, and Y.~Pan, ``A study of the routing and spectrum allocation
  in spectrum-sliced elastic optical path networks,'' in \emph{Proc. of INFOCOM
  2011}, Apr. 2011, pp. 1503--1511.

\bibitem{s3:b1}
M.~Gr{\"o}tschel, L.~Lov{\'a}sz, and A.~Schrijver, \emph{Geometric Algorithms
  and Combinatorial Optimization}.\hskip 1em plus 0.5em minus 0.4em\relax
  Springer, 1988, vol.~2.

\bibitem{s3:b2}
F.~Gavril, ``The intersection graphs of subtrees in trees are exactly the
  chordal graphs,'' \emph{J. Combinat. Theory}, vol.~16, pp. 47--56, Jan. 1974.

\bibitem{s2:b3}
U.~Feige and J.~Kilian, ``Zero knowledge and the chromatic number,'' \emph{J.
  Comput. Syst. Sci.}, vol.~57, pp. 187--199, Feb. 1998.

\bibitem{s3:b4}
M.~Halldorsson, ``A still better performance guarantee for approximate graph
  coloring,'' \emph{Inform. Proc. Lett.}, vol.~45, pp. 19--23, Jan. 1993.

\bibitem{s3:b5}
Bollobas, ``The chromatic number of random graphs,'' \emph{Combinat.}, vol.~8,
  pp. 49--55, 1988.

\bibitem{ss4:b5}
S.~Cho and S.~Ramasubramanian, ``Localizing link failures in all-optical
  networks using monitoring tours,'' \emph{Computer Networks}, vol.~58, no.~17,
  pp. 2--12, 2014.

\bibitem{ss4:b3}
X.~Dong, T.~El-Gorashi, and J.~Elmirghani, ``On the energy efficiency of
  physical topology design for ip over wdm networks,'' \emph{J. Lightw.
  Technol.}, vol.~30, no.~12, pp. 1931--1942, Jun. 2012.

\bibitem{ss4:b4}
M.~Ju, F.~Zhou, S.~Xiao, and Z.~Zhu, ``Power-efficient protection with directed
  $p$ -cycles for asymmetric traffic in elastic optical networks,'' \emph{J.
  Lightw. Technol.}, vol.~34, no.~17, pp. 4053--4065, Sept. 2016.

\end{thebibliography}
\end{document}